\documentclass[twocolumn,envcountsect,envcountsame]{svjour3}    
\usepackage[utf8]{inputenc}

\usepackage{graphicx}
\usepackage{natbib}
\usepackage[colorlinks,linkcolor=blue,citecolor=blue,urlcolor=blue]{hyperref}
\usepackage{times}
\usepackage[mathscr]{eucal}
\usepackage{amsmath,amsfonts,amssymb,bbm,bm}
\usepackage{graphicx,epsfig,subfigure,cases,multirow}
\usepackage[linesnumbered,ruled,vlined,commentsnumbered]{algorithm2e}
\usepackage{enumerate}
\usepackage[title]{appendix}
\usepackage[displaymath,mathlines,switch]{lineno}
\usepackage{booktabs}
\usepackage[capitalize]{cleveref}
\usepackage{geometry}
\def\cl@chapter{\@elt {chapter}}
\makeatother
\usepackage[capitalize]{cleveref}

\makeatletter
\newcommand{\labelx}[1]{
    \relax
    \ifmmode
        \label{#1} 
    \else 
        \ifnum\pdfstrcmp{\@currenvir}{document}=0
            \label{#1}
        \else
            \label[\@currenvir]{#1}
        \fi
    \fi
}
\makeatother

\newlength\savewidth

\newlength\savedwidth

\newcommand{\bE}{\mathbb E}
\newcommand{\bS}{\mathbb S}
\newcommand{\bR}{\mathbb R}
\newcommand{\bZ}{\mathbb Z}
\newcommand{\bM}{\mathbb M}

\newcommand{\OLP}{\bm{\rho}}

\newcommand{\ks}{\mathfrak s}
\newcommand{\ky}{\mathfrak y}
\newcommand{\kB}{\mathfrak B}
\newcommand{\kL}{\mathfrak L}
\newcommand{\kF}{\mathfrak F}
\newcommand{\kH}{\mathfrak H}

\newcommand{\kT}{\mathfrak T}

\newcommand{\cS}{\mathcal S}

\newcommand{\cG}{\mathcal G}
\newcommand{\cH}{\mathcal H}
\newcommand{\cL}{\mathcal L}

\newcommand{\cF}{\mathcal F}
\newcommand{\cB}{\mathcal B}
\newcommand{\cT}{\mathcal T}
\newcommand{\cO}{\mathcal O}
\newcommand{\cD}{\mathcal D}

\newcommand{\cN}{\mathcal N}
\newcommand{\cV}{\mathcal V}
\newcommand{\tka}{\tilde\kappa}

\def\<{\langle} 
\def\>{\rangle}
\def\sm{\setminus}

\newcommand{\vp}{\varphi}
\newcommand{\ve}{\varepsilon}

\newcommand{\fs}{\mathbf s}
\newcommand{\fx}{\mathbf x}

\newcommand{\fy}{\mathbf y}

\newcommand{\dkq}{{\dot{\mathfrak q}}}
\newcommand{\dn}{\dot{\mathrm{n}}}
\newcommand{\dx}{\dot x}
\newcommand{\dfx}{\dot{\mathbf x}}
\newcommand{\dfn}{{\dot{\mathbf{n}}}}
\newcommand{\dtheta}{\dot\theta}
\newcommand{\htheta}{\hat\theta}
\newcommand{\hfx}{\hat{\mathbf{x}}}

\newcommand{\hx}{\hat{x}}

\newcommand{\dfv}{{\dot{\mathbf v}}}
\newcommand{\dfe}{\dot{\mathbf e}}

\newcommand{\curPr}{\omega}

\DeclareMathOperator\sign{sign}

\DeclareMathOperator\Lip{Lip}
\DeclareMathOperator\JS{JS}

\journalname{Preprint}
\begin{document}

\title{Computing Geodesic Paths Encoding a Curvature Prior for Curvilinear Structure Tracking
}


\author{Da~Chen \and Jean-Marie~Mirebeau \and Minglei Shu \and Laurent~D.~Cohen
}


\institute{Da~Chen\at
Shandong Artificial Intelligence Institute, Qilu University of Technology (Shandong Academy of Sciences), 250014 Jinan, China.\\
\email{dachen.cn@hotmail.com}
\and Jean-Marie Mirebeau \at
Department of Mathematics, Centre Borelli, ENS Paris-Saclay, CNRS, University Paris-Saclay, 91190, Gif-sur-Yvette, France.\\
\email{jean-marie.mirebeau@math.u-psud.fr}\\
\and Minglei Shu \at
Shandong Artificial Intelligence Institute, Qilu University of Technology (Shandong Academy of Sciences), 250014 Jinan, China.\\
\email{shuml@sdas.org}\\   
\and Laurent~D.~Cohen \at
University Paris Dauphine, PSL Research University, CNRS, UMR 7534, CEREMADE, 75016 Paris, France\\ 
\email{cohen@ceremade.dauphine.fr}\\   
}

\date{Received: date / Accepted: date}

\maketitle
\begin{abstract}
In this paper, we introduce an efficient method for computing curves minimizing a variant of the Euler-Mumford elastica energy, with fixed endpoints and tangents at these endpoints, where the bending energy is enhanced with a user defined and data-driven scalar-valued term referred to as the curvature prior.
In order to guarantee that the globally optimal curve is extracted, the proposed method involves the numerical computation of the viscosity solution to a specific static Hamilton-Jacobi-Bellman (HJB) partial differential equation (PDE). For that purpose, we derive the explicit Hamiltonian associated to this variant model equipped with a curvature prior, discretize the resulting HJB PDE using an adaptive finite difference scheme, and  solve it in a single pass using a generalized Fast-Marching method. In addition, we also present a practical method for estimating the curvature prior values from image data, designed for the task of accurately tracking curvilinear structure centerlines. Numerical experiments on synthetic and real image data illustrate the advantages of the considered variant of the elastica model with a prior curvature enhancement in complex scenarios where challenging geometric structures appear. 
\end{abstract}

\section{Introduction}
Computing globally optimal geodesic curves, in domains equipped with complex metric cost functions, is a task efficiently achieved by solving the so-called eikonal equation, which is at the foundation of a line of works in the modeling of curvilinear shapes \citep{sethian1999fast,peyre2010geodesic}.
Second-order geodesic models featuring curvature regularization allow imposing particular geometric priors (e.g.\ rigidity and smoothness) to the minimal geodesic paths, whose computation hence has become a core step to numerous practical applications. Typical examples include image analysis~\citep{chen2017global,bekkers2015pde}, computer vision problems involving  visual curve completion~\citep{ben2014tangent,parent1989trace} and perceptual grouping~\citep{bekkers2018nilpotent,wertheimer1938laws}, the study of optical illusions~\citep{franceschiello2019geometrical},  and robot motion planing~\citep{mirebeau2019hamiltonian,kimmel2001optimal}.
This work is motivated by the fact that curvilinear structures exhibit significant geometric features such as orientation and curvature, which can be leveraged as crucial cues for computing geodesic paths to efficiently delineate these elongated shapes.
 
Approaches to computing globally optimal paths connecting two disjoint sets can be backtracked to Dijkstra's shortest path algorithm~\citep{dijkstra1959note} in a finite graph.  However, such a discrete setting may lead to metrication errors, especially for applications in image analysis. In order to address this issue,  Cohen and Kimmel introduced a minimal geodesic model~\citep{cohen1997global} founded in a continuous PDE domain, and established the connection between the minimization of an energy defined as a weighted curve length~\citep{caselles1997geodesic} and the viscosity solution to a static first-order HJB PDE. As a consequence, the computation of minimal geodesic paths is transferred to numerically solving the HJB PDE, which is efficiently addressed via Sethian's pioneering work~\citep{sethian1996fast,sethian1999fast}  the Fast-Marching method. 

The original geodesic model~\citep{cohen1997global} measures the length of regular curves in an isotropic manner, using weights depending on the local curve position but not on its orientation, which limits its practical applications. 
The introduction of anisotropic geometric models, and measures of path length depending on the orientation in a user-defined manner, is an important extension to the original geodesic model, 
which has deeply extended its applicable scope, leading to successful applications in image analysis and computer vision~\citep{melonakos2008finsler}.
For instance, curvilinear structure tracking can take advantage of anisotropic Riemannian models, by enhancing the geodesic metric with anisotropy features extracted from the elongated structures present in the processed image~\citep{parker2002estimating,benmansour2011tubular,jbabdi2008accurate}.
\cite{chen2016finsler} proposed to address the region-based active contour problem~\citep{mumford1989optimal} by computing the minimal geodesic path associated to a Randers metric~\citep{randers1941asymmetrical} featuring a data-driven drift term. 
For geodesic models involving anisotropic metrics, the classical Fast-Marching method~\citep{sethian1996fast} is no longer suitable for estimating the solutions to the respective HJB PDEs~\citep{sethian2001ordered,sethian2003ordered}. In order to bridge this gap, Mirebeau introduced variants of the Fast-Marching method which rely on adaptive stencil systems depending on the associated metrics, based on either the geometric tool of Lattice basis reduction~\citep{mirebeau2014efficient,mirebeau2014anisotropic} or the Voronoi's first reduction technique~\citep{mirebeau2019riemannian,mirebeau2018fast}. 
In practice, those state-of-the-art generalized anisotropic Fast-Marching methods can address various anisotropic geodesic models, while achieving sufficient accuracy and requiring a low computational cost.
 
\begin{figure*}[t]
\centering
\subfigure[]{\includegraphics[width=0.32\linewidth]{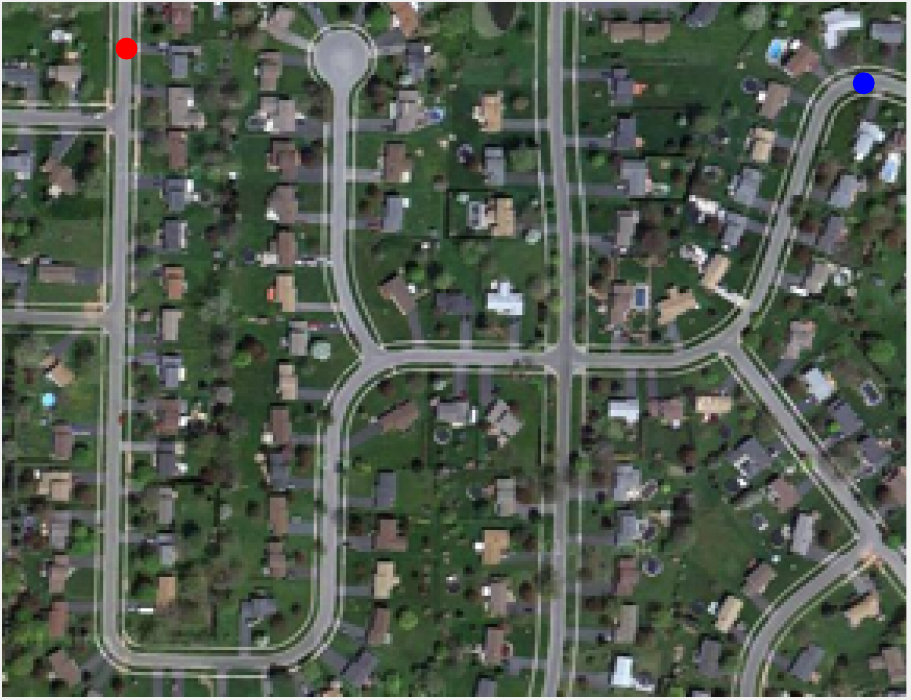}}\,\subfigure[]{\includegraphics[width=0.32\linewidth]{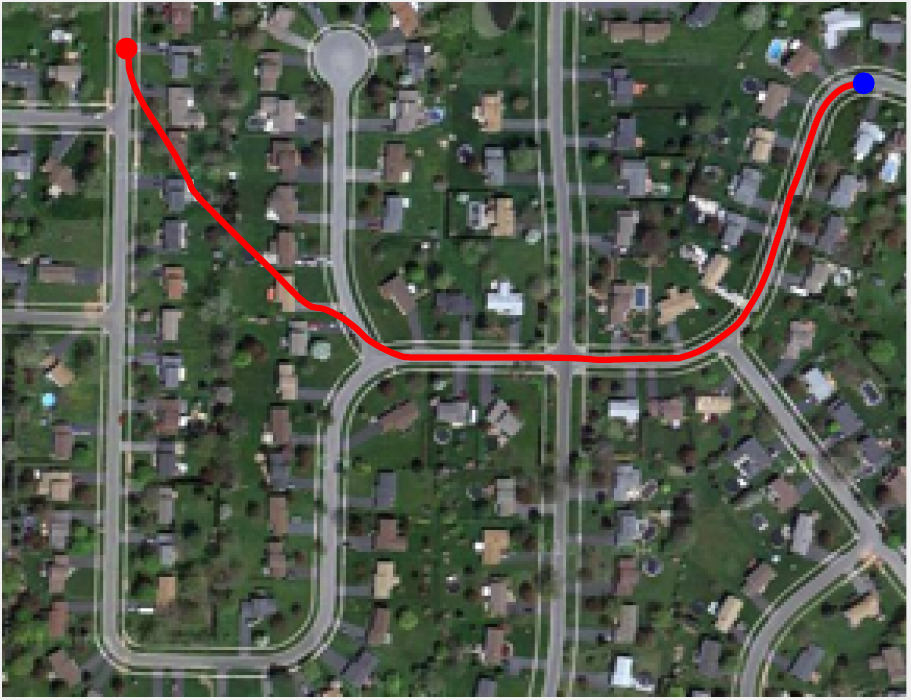}}\,\subfigure[]{\includegraphics[width=0.32\linewidth]{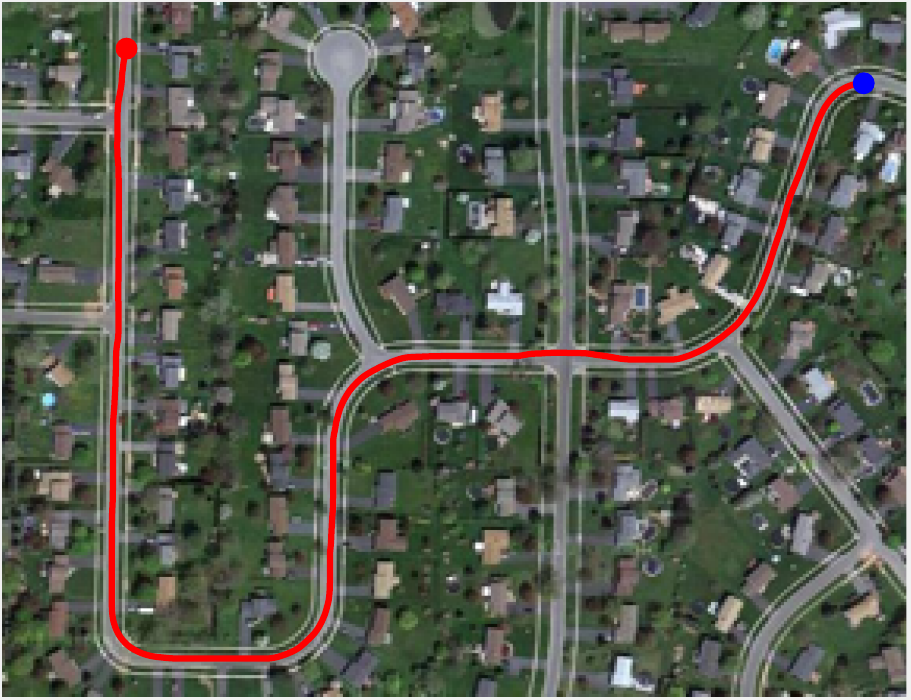}}
\caption{Illustration of the extraction of geodesic paths from an aerial image.
\textbf{a}  The original image, where the red and  blue dots represent the two endpoints of the target structure. \textbf{b} and \textbf{c} Physical projections of the minimal geodesic paths for the Euler-Mumford elastica model and the variant of the elastica model with a curvature prior enhancement respectively.}
\label{fig_CompRoad}
\end{figure*}

Second-order geodesic models, which use the path curvature as a regularizer of the weighted curve length, can be regarded as limit cases of anisotropic models defined over the configuration space of positions and orientations, see~\citep{chen2017global,duits2018optimal,bekkers2015pde,mirebeau2018fast} for significant examples. Following the optimal control framework, the Euler-Mumford elastica, Reeds–Shepp car and Dubins car problems are respectively connected to the associated HJB PDEs, yielding efficient Fast-Marching schemes for finding curvature-penalized optimal curves~\citep{mirebeau2018fast}. An extension of these models allows computing curves which are boundaries of convex shapes, by controlling the sign of the curvature and the winding number of the path~\citep{chen2021elastica}. In these works, the curvature penalization is regarded as a regularization term, whose importance is delicately balanced with an image data term, following the general principles of data fitting. 

In practice, unfortunately, striking a good balance between the curvature regularization and the image data terms remains a challenging problem for these second-order geodesic  models~\citep{chen2017global,duits2018optimal,mirebeau2018fast}, particularly in the presence of crossing structures, strongly bent segments, and complex image background content. In general, assigning high importance weights to curvature regularization usually corresponds to optimal paths favoring low curvature. On the other hand, such a setting may in turn introduce a bias to  the underlying strongly curved segments, thus giving rise to unexpected results. Fig.~\ref{fig_CompRoad} illustrates an example in an aerial image, where the goal is to extract a road joining two given points (red and blue dots), see Fig.~\ref{fig_CompRoad}a. In Fig.~\ref{fig_CompRoad}b, one can see that the geodesic path derived from the classical Euler-Mumford elastica model~\citep{chen2017global,mirebeau2018fast} suffers from a shortcut problem,  due to the unbalanced importance of the weights associated to the curvature regularization and to the image data. This shortcoming of the classical elastica model prevents its practical applications in complex scenarios. An alternative approach, explored in this paper and motivated by the drawback as mentioned above, is to compute geodesic paths whose curvature approximates some data-driven prior reference values and more generally to introduce data information in the new curvature term. 
An additional contribution is a collection of techniques to estimate these prior curvature values from the image data, and to extend and preprocess them for their use as an  enhancement to the geodesic model. Fig.~\ref{fig_CompRoad}c illustrates an optimal path from the variant of the elastica model, equipped with a suitable curvature prior which enables the accurate extraction of the target road.

In~\citep{mirebeau2019hamiltonian}, the authors illustrated how geodesic paths with a type of curvature prior could be computed using a modification of the Hamiltonian Fast-Marching (HFM) method~\citep{mirebeau2018fast,mirebeau2019riemannian}. However, the mathematical PDE framework, the construction of the numerical scheme, and the application scope of the geodesic models with embedded curvature prior were not investigated. In this paper, we describe the Hamiltonian of a variant of the Euler-Mumford elastica model featuring a curvature prior enhancement, and we illustrate related objects known as control sets. 
The Hamiltonian of this variant of the classical elastica model, referred to as the \emph{curvature prior elastica model}, is then discretized, leading to a numerical method for solving the corresponding static first-order HJB PDE, and thus computing the globally optimal geodesic paths. Finally, we introduce a practical method for estimating the curvature prior values, designed to improve the performance of tracking curvilinear structures from challenging image data. 

The remainder of this manuscript is organized as follows. We first review the classical Euler-Mumford elastica model, including the expression of the associated HJB PDE which is posed over the configuration space of positions and orientations,  and the computation of the relevant minimal geodesic paths. Then we present the core of this work: the derivation of the Hamiltonian of a variant of the elastica model equipped with a curvature prior, the relevant static HJB PDE,  and its implementation in the application of curvilinear structure tracking. Eventually, we present experimental results conducted on both synthetic and real images.

\section{The Euler-Mumford  Elastica Model}
\subsection{Background}
The celebrated Euler-Mumford elastica model \citep{mumford1994elastica,nitzberg19902} characterizes smooth curves which minimize a functional featuring a length term and a bending energy term, in other words a curvature penalization. 
Let  $\kappa:[0,1]\to\bR$ be the curvature of a smooth curve $\gamma:[0,1]\to\Omega$ with non-vanishing velocity, where $\Omega\subset\bR^2$ is a bounded domain. Given a data-driven cost function $\tilde\psi:\Omega\times\bR^2\to]0,\infty[$, the energy functional considered in the classical Euler-Mumford elastica model reads 
\begin{equation}
\label{eq_ElasticaEnergy} 
\int_0^1\tilde\psi\left(\gamma(u),\gamma^\prime(u)/\|\gamma^\prime(u)\|\right)(1+\beta^2\kappa(u)^2)\|\gamma^\prime(u)\|du,
\end{equation}
where $\gamma^\prime$ is the first-order derivative of the curve $\gamma$, and $\|\cdot\|$ denotes the standard Euclidean norm over the space $\bR^2$. 
The constant parameter $\beta>0$ controls the relative importance of the bending energy term $\kappa(u)^2$, which penalizes curvature, w.r.t.\ the penalization of path length. By construction, \cref{eq_ElasticaEnergy} is invariant under smooth increasing re-parametrizations of the curve $\gamma$, and the second argument of $\tilde \psi$ is the unit tangent  vector to the path. In the special case where $\tilde\psi$ is constant, the curves $\gamma$ for which \cref{eq_ElasticaEnergy} is extremal correspond to the rest positions of elastic rods of suitable length (with their endpoints clamped), however this physical interpretation is lost for non-constant $\tilde \psi$. 

\begin{figure*}[t]
\centering
\subfigure[]{\includegraphics[height=0.28\linewidth]{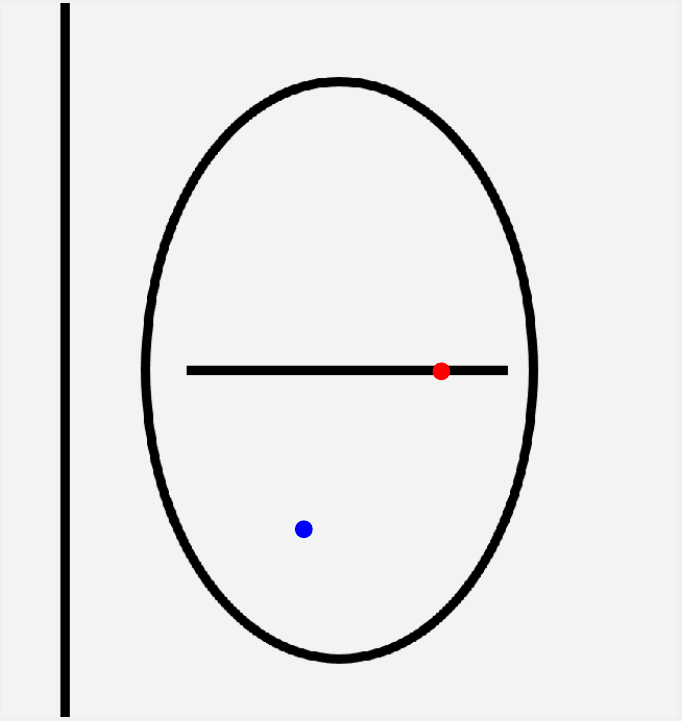}}\qquad\subfigure[]{\includegraphics[height=0.28\linewidth]{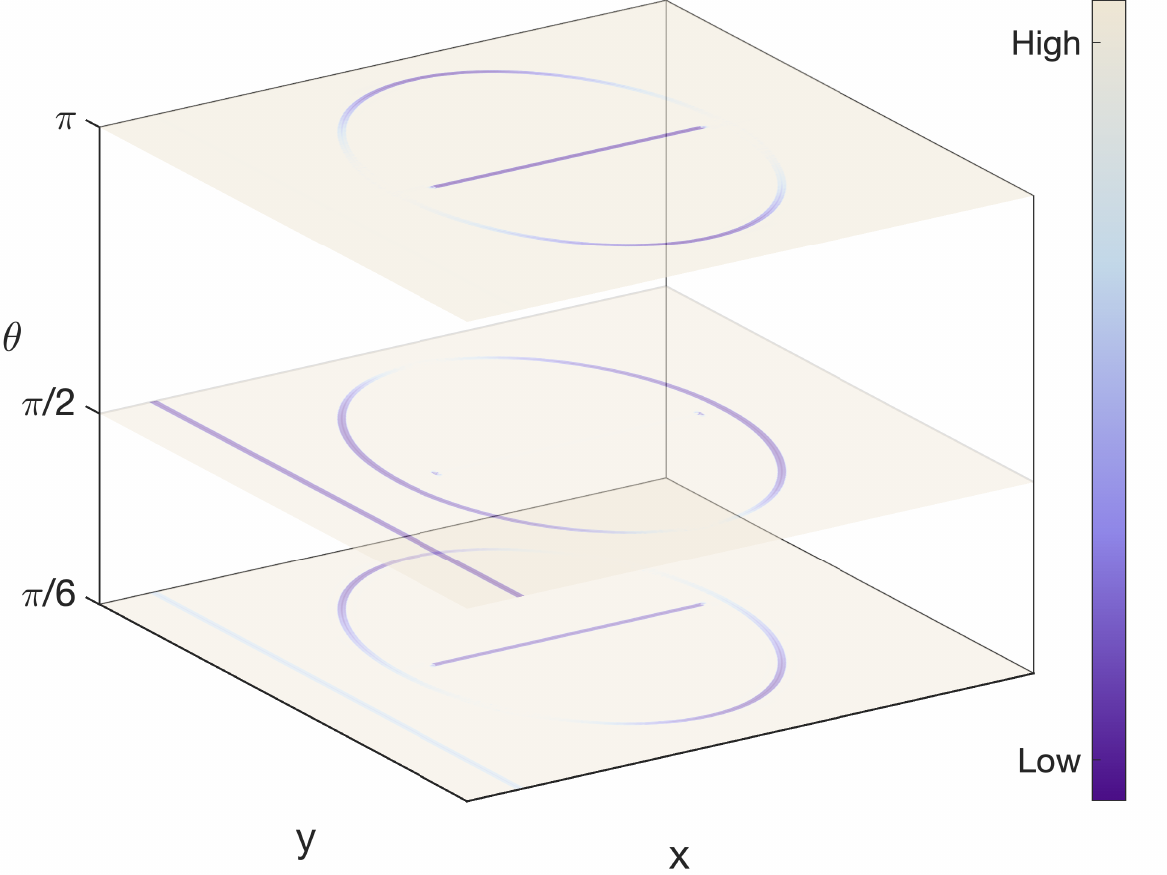}}\qquad\subfigure[]{\includegraphics[height=0.28\linewidth]{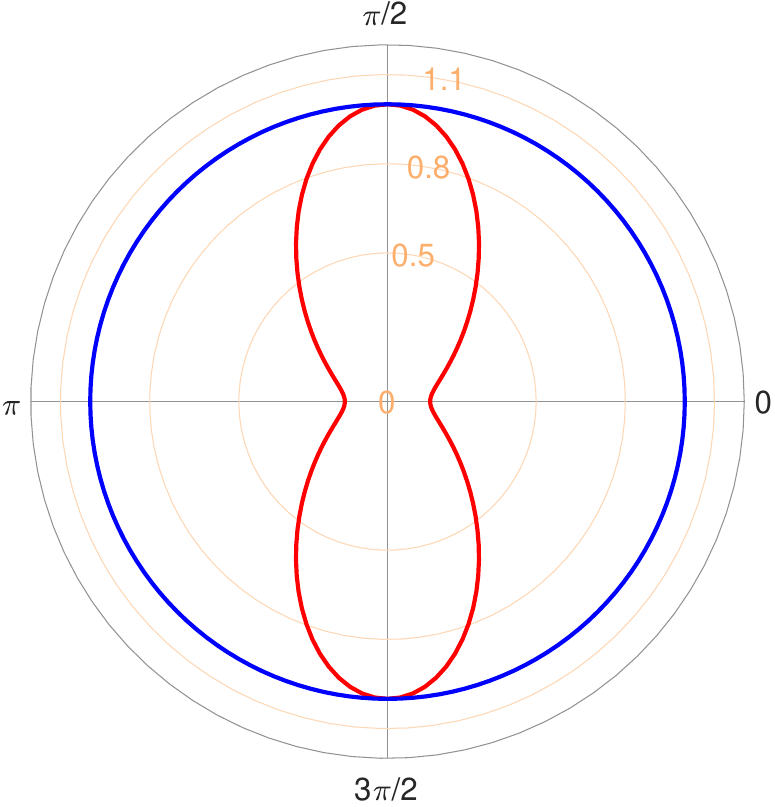}}
\caption{Visualizing the cost function $\psi$ in a synthetic example. \textbf{a} The original image consisting of curvilinear structures, where the red and blue dots indicate two sampled points. 
\textbf{b} Visualizing the cost function $\psi(\cdot,\theta)$ at three slices corresponding to the angles $\theta=\pi/6$, $\pi/2$ and $\pi$. 
\textbf{c} Visualizing the cost function $\psi(x,\cdot)$ at two sampled points in polar coordinates, where the red and blue lines correspond to plots of the values at the red and blue dots, respectively.
}	
\label{fig_OS}
\end{figure*}

\subsection{Reformulation via Orientation Lifting}
A widely used strategy for addressing the classical Euler-Mumford elastica optimal curve problem is to lift a planar smooth curve $\gamma$ into the higher dimensional configuration space $\bM:=\Omega\times\bS^1$ of positions and orientations, see~\citep{chen2017global,duits2018optimal,chambolle2019total}, where $\bS^1:=\bR/2\pi\bZ$ denotes the angular domain. In other words $\bS^1 = [0,2\pi[$ with periodic boundary conditions, and we can parametrize the unit circle via the unit vector $\dn(\theta):=(\cos\theta,\sin\theta)$ for any angle $\theta\in\bS^1$.
Furthermore, we denote by $\bE:=\bR^2\times\bR$ the tangent vector space to $\bM$, and by $\bE^*$ the dual vector space of $\bE$. 
Recall that $\gamma : [0,1] \to \Omega$ is a smooth curve with non-vanishing velocity.
Following the literature~\citep{chen2017global,mirebeau2018fast,duits2018optimal}, an orientation-lifted curve $\OLP=(\gamma,\eta):[0,1]\to\bM$, whose \emph{physical projection} is the given planar curve $\gamma$, can be constructed by introducing an \emph{angular function} $\eta : [0,1] \to \bS^1$ such that 
\begin{equation}
\label{eq_TangentCharacteristic}
\gamma^\prime(u)=\dn(\eta(u))\|\gamma^\prime(u)\|,\quad \forall u\in[0,1].
\end{equation}
This leads to an alternative representation of the curvature
\begin{equation}
\label{eq_Curvature}
\kappa(u)=\frac{\eta^\prime(u)}{\|\gamma^\prime(u)\|}.
\end{equation}
Inserting this expression of the curvature $\kappa$ into the energy defined in~\cref{eq_ElasticaEnergy} yields an equivalent formula, measured along an orientation-lifted curve $\OLP=(\gamma,\eta)$ whose components $\gamma$ and $\eta$ obey~\cref{eq_TangentCharacteristic}:
\begin{equation}
\label{eq_OLElasticaEnergy}
\cL(\OLP):=\int_0^1\psi(\gamma(u),\eta(u))\left(\|\gamma^\prime(u)\|+\frac{\beta^2\eta^\prime(u)^2}{\|\gamma^\prime(u)\|}\right)du,
\end{equation}
where $\psi:\bM\to]0,\infty[$ is a positive cost function defined by
\begin{equation*}
\psi(x,\theta) := \tilde\psi(x,\dn(\theta)).	
\end{equation*}
For the task of curvilinear structure tracking, as considered in this work, the value of  $\psi(x,\theta)$ should be relatively low in case the physical position $x$ is inside a curvilinear structure and  the direction $\dn(\theta)$ matches its local geometry at the position $x$, i.e.\ $\dn(\theta)$ is approximately tangent to its centerline at $x$.  Fig.~\ref{fig_OS} shows an example for the cost function $\psi$ derived from an image that consists of curvilinear structures. The construction and the numerical computation of $\psi$ is presented in~\cref{eq_CostFunction} below, and relies on a tool known as the orientation score.

\subsection{Elastica Metric and Control Sets}
The energy $\cL(\OLP)$ formulated in~\cref{eq_OLElasticaEnergy} can be expressed in terms of a degenerate geodesic metric~\citep{chen2017global,mirebeau2018fast}, denoted by $\cF:\bM\times\bE\to[0,+\infty]$ and referred to as the \emph{elastica metric}. For any point $\fx=(x,\theta)$ and any non-zero vector $\dot\fx=(\dx,\dtheta)$ this elastica metric is defined as 
\begin{equation}
\label{eq_ElasticaMetric}
\cF(\fx,\dot\fx):=
\begin{cases}
\|\dx\|+\displaystyle\frac{(\beta\dtheta)^2}{\|\dx\|}, &\text{if~}\dx=\dn(\theta)\|\dx\|,\\
+\infty,&\text{otherwise}.
\end{cases}
\end{equation}
Note that this metric is not symmetric, since $\cF(\fx,\dot \fx) \neq \cF(\fx,-\dot \fx)$ in general in view of the non-linear constraint $\dx=\dn(\theta)\|\dx\|$, and thus the corresponding distance defined below (see~\cref{eq_GDM}) is likewise non-symmetric. Such asymmetric objects are often referred to as quasi-metrics and quasi-distances, but the prefix ``quasi-'' is dropped in this paper for the sake of readability. 

Several mathematical objects can be attached to such a metric. 
The \emph{control sets} $\cB(\fx) \subset \bE$, for all points $\fx \in \bM$, allow visualizing the geometry, see Fig.~\ref{fig_ControlSet}, and are defined as 
\begin{equation}
\label{eqdef:control_set}
	\cB(\fx) := \left\{\dfx\in \bE \mid \cF(\fx,\dfx) \leq 1\right\}.
\end{equation}
The \emph{Hamiltonian} $\cH :\bM\times\bE^*\to [0,\infty[$ is involved in the numerical computation of optimal geodesic paths through the HJB PDE, as formulated in next subsection, and is defined for any co-vector $\hfx \in \bE^*$ as
\begin{equation}
\label{eqdef:hamiltonian}
\cH(\fx,\hfx) := \frac 1 2 \Big( \max_{\dfx \in \cB(\fx)} \<\hfx,\dfx\> \Big)^2. 
\end{equation} 

In the special case of the elastica metric, one can use \cref{eq_ElasticaMetric} to obtain the explicit expressions of the control sets and Hamiltonian, see \citep{mirebeau2018fast} for the details of this computation. Specifically, for any point $\fx=(x,\theta)\in\bM$ one has 
\begin{align}
\label{eq:elastica_control_set}
	&\cB(\fx) =\nonumber \\
	&\left\{(\dx,\dtheta) \in \bE \mid \dx = \dn(\theta)\|\dx\|, \left(\|\dx\|-\frac{1}{2}\right)^2 + \beta^2 \dtheta^2\leq\frac{1}{4}\right\}.
\end{align}
For any co-vector $\hfx=(\hx,\htheta)\in\bE^*$ one obtains
\begin{equation}
\label{eq_ElasticaHam1}	
\cH(\fx,\hfx)=
\frac{1}{8}\left(\<\hx,\dn(\theta) \>+\sqrt{\<\hx,\dn(\theta)\>^2+(\htheta/\beta)^2}\right)^2.
\end{equation}
The following equivalent expression of the Hamiltonian, in integral form, is established in \cite{mirebeau2018fast} 
\begin{equation}
\label{eq_ElasticaHam2}	
\cH(\fx,\hfx)=\frac{3}{8}\int_{-\pi/2}^{\pi/2}\<\hfx,(\dn(\theta)\cos\varphi,\beta^{-1}\sin\varphi)\>^2_+\cos\varphi\,d\varphi,
\end{equation}
where $\< \cdot,\cdot\>_+:=\max\{0,\<\cdot,\cdot \>\}$ is the positive part of the Euclidean scalar product $\<\cdot,\cdot\>$.
Here and below we denote the vector $\dfn(\theta,a):=(\dn(\theta),a)\in\bE$ for any angle $\theta \in \bS^1$ and any scalar $a \in \bR$, and observe that in particular $\<\hx,\dn(\theta)\> = \<\hfx,\dfn(\theta,0)\>$.

\subsection{Computing Elastica Geodesic Paths}
Consider a fixed source point $\fs\in\bM$ and an arbitrary target point $\fx\in\bM$, thus defining initial and final planar endpoints and  tangent path directions. The classical Euler-Mumford elastica geodesic model~\citep{chen2017global,mirebeau2018fast} aims to track a minimal geodesic path  $\cG_{\fs,\fx}:[0,1]\to\bM$ which links from the source point $\fs$ to the target point $\fx$ (i.e.\ $\cG_{\fs,\fx}(0)=\fs$, $\cG_{\fs,\fx}(1)=\fx$) and minimizes the energy $\cL$ formulated in~\cref{eq_OLElasticaEnergy}. For that purpose, we first consider the set $\Lip([0,1],\bM)$ of all Lipschitz continuous curves $\OLP:[0,1]\to\bM$ as the search space of geodesic paths. The minimization of the energy $\cL$ yields a geodesic distance map $\cD_\fs:\bM\to[0,\infty)$, defined as
\begin{equation}
\label{eq_GDM}	
\cD_\fs(\fx):=\inf_{\OLP\in\Lip([0,1],\bM)}\,\cL(\OLP),\quad s.t.~
\begin{cases}
\OLP(0)=\fs,&\\
\OLP(1)=\fx.&
\end{cases}
\end{equation}
The geodesic distance map $\cD_\fs$ is a potentially discontinuous viscosity solution to the first-order static HJB PDE~\citep{Bardi:2008Viscosity}
\begin{equation}
\label{eq_HJB}
\begin{cases}
\displaystyle\cH(\fx,d\cD_\fs(\fx))=\frac{1}{2}\psi(\fx)^2,~\forall \fx\in\bM\backslash\{\fs\},\\
\cD_\fs(\fs)=0,	
\end{cases}
\end{equation}
where $d\cD_\fs$ denotes the differential of the geodesic distance map $\cD_{\fs}$. Outflow boundary conditions are applied on $\partial \bM$.

A \emph{reversely parametrized} minimal geodesic path $\tilde \cG$, starting from a given target point $\fx$ and going back to the source point $\fs$, can be obtained by solving the following gradient descent ordinary differential equation (ODE), known as the \emph{geodesic backtracking} ODE
\begin{equation}
\label{eq_ODE}
\tilde\cG^\prime(u)= - \partial_{2}\cH(\tilde\cG(u),d\cD_{\fs}(\tilde\cG(u))),\quad 	
\end{equation}
with initial condition $\tilde\cG(0)=\fx$, and where $\partial_2\cH:=\partial\cH(\fx,\hfx)/\partial\hfx$. 
By construction, this path is well-defined up to the minimal arrival time $T:=\cD_\fs(\fx)$, and $\tilde \cG(T) = \fs$ is the source point.
The reparametrization $\cG_{\fs,\fx}(u) := \tilde \cG(T(1-u))$ yields a minimal geodesic path $\cG_{\fs,\fx} : [0,1] \to \bM$ such that $\cG_{\fs,\fx}(0)=\fs$ and $\cG_{\fs,\fx}(1)=\fx$.

\section{Curvature Prior Elastica Model: A Variant of the Classical Euler-Mumford Elastica Model}
In this section, we present a variant of the Euler-Mumford elastica geodesic model, favoring minimal geodesic paths whose curvature remains close to a prior reference value, defined pointwise. For comparison, the classical Euler-Mumford elastica model~\citep{chen2017global}, which favors geodesic paths of low curvature, corresponds to the special case where these prior values are defined as zero uniformly over the domain $\bM$.

\subsection{Formalism}
We begin by formulating a new curvature penalization term, which involves a scalar-valued function $\curPr:\bM\to\bR$. For a smooth curve $\gamma$ with curvature $\kappa$, we consider
\begin{equation}
\label{eq_PriorCurvature}
\tka(u)=\kappa(u)-\curPr\big(\gamma(u),\eta(u)\big),	
\end{equation}
where the angle $\eta(u)\in \bS^1$ reflects the tangent path direction for all $u\in[0,1]$, see~\cref{eq_TangentCharacteristic}. 
In this model, the function $\omega$ is referred to as the curvature prior map. 
For the purposes of curvilinear structure tracking, as considered in this paper the value $\omega(x,\theta)$ may be defined as the curvature of a reference curve passing nearby the position $x$, with a tangent orientation close to $\theta$, and obtained using an independent skeletonization procedure, see \cref{eq_PriorConstruction}.

We consider the variant of the Euler-Mumford elastic energy involving the curvature prior map $\omega$ and defined as:
\begin{equation}
\label{eq_PriorOriginalEnergy}
\int_0^1 \psi(\gamma(u),\eta(u))\left(1+(\beta\tka(u))^2\right)\|\gamma^\prime(u)\|du.
\end{equation}
For an orientation-lifted curve $\OLP=(\gamma,\eta)$ satisfying~\cref{eq_TangentCharacteristic}, the corresponding quantity is defined as 
\begin{equation}
\label{eq_VariantLiftedEnergy}
\kL(\OLP):=\int_0^1\psi(\OLP)\left(1+\beta^2\left(\frac{\eta^\prime}{\|\gamma^\prime\|}-\curPr(\OLP)\right)^2\right)\|\gamma^\prime\|du.
\end{equation} 
The energy $\kL$ leads to a new metric $\kF:\bM\times\bE\to[0,\infty]$~\citep{mirebeau2019hamiltonian}
\begin{equation}
\label{eq_PriorMetric}
\kF(\fx,\dfx):=
\begin{cases}
\|\dx\|+\frac{\beta^2(\dtheta-\curPr(\fx)\|\dx\|)^2}{\|\dx\|}, &\text{if~}\dx=\dn(\theta)\|\dx\|,\\
\infty,&\text{otherwise},
\end{cases}
\end{equation}
where $\fx=(x,\theta)\in\bM$, and where $\dfx=(\dx,\dtheta)\in\bE$, with the convention $\kF(\fx,\mathbf{0}) = 0$. 
The energy $\kL(\OLP)$ can therefore be reformulated as 
\begin{equation*}
\kL(\OLP)=\int_0^1\psi(\OLP(u))\kF(\OLP(u),\OLP^\prime(u))du.
\end{equation*}
Path globally minimizing this energy can be obtained, similarly to the classical case and following the general principles of optimal control, by applying the geodesic backtracking ODE to the viscosity solution of a HJB PDE, both involving a suitably modified Hamiltonian described in the next section. 

\begin{figure*}[t]
\centering
\includegraphics[height=0.28\linewidth]{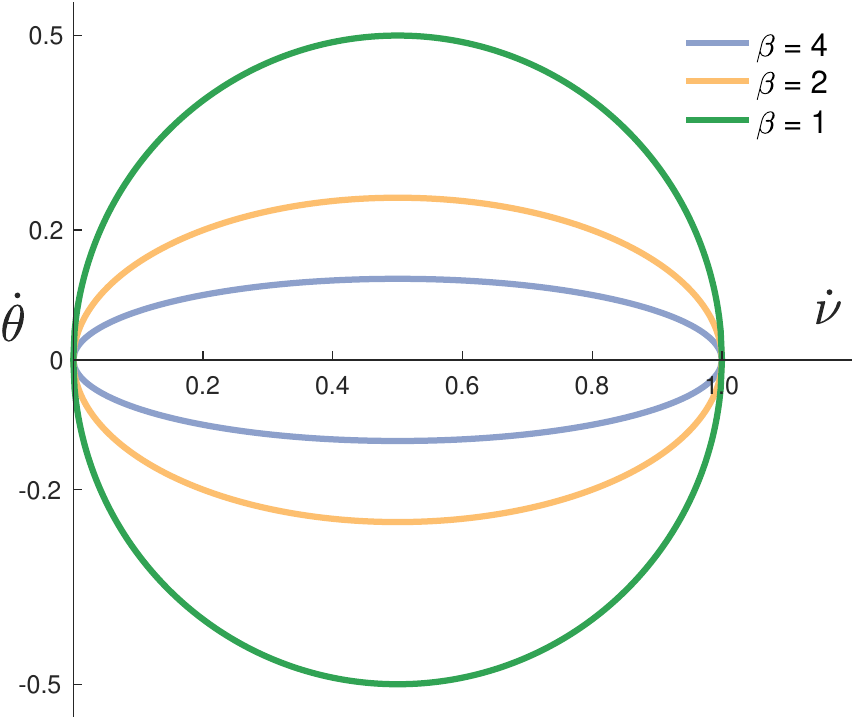}\quad\includegraphics[height=0.28\linewidth]{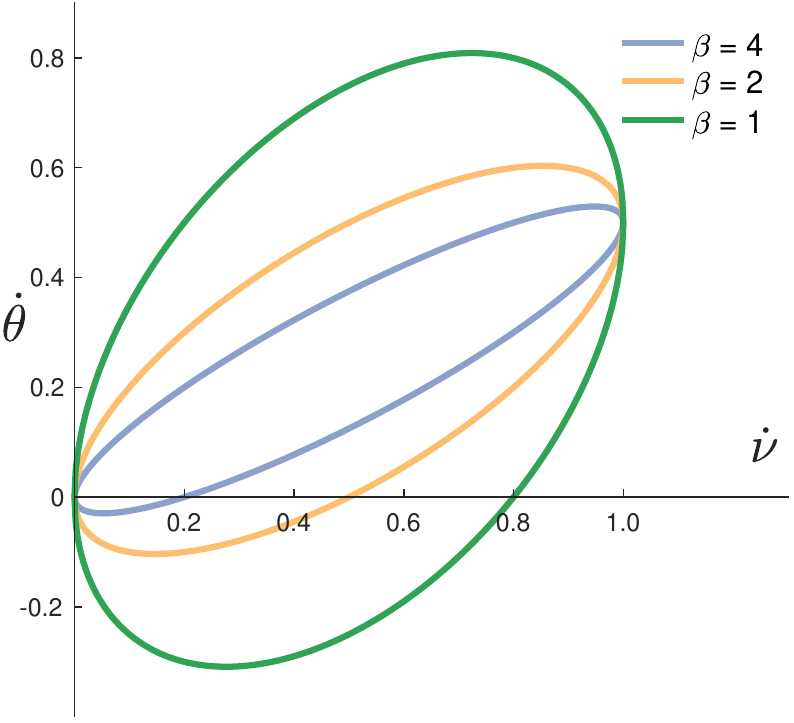}\quad\includegraphics[height=0.28\linewidth]{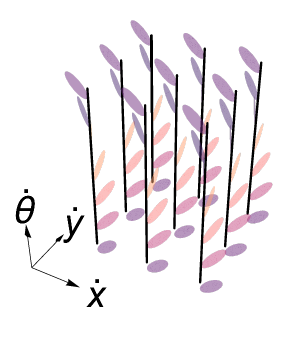}
\caption{Illustration of the control sets associated with the classical elastica model and the curvature prior elastica model. \textbf{Left}: The lines of different colors denote the sets of all $(\dot\nu,\dot\theta)$ such that $\cF(\fx,(\dot\nu\dn(\theta),\dot\theta))=1$ w.r.t. curvature penalization parameter $\beta\in\{1,2,4\}$, where $\dot\nu:=\|\dx\|$. Each of these sets is the boundary of the intersection between the control set $\cB(\fx)$ of  the elastica metric $\cF$ and the slice defined by $\dot x = \dn(\theta) \|\dot x\|$. \textbf{Center}: Likewise for the curvature prior  elastica model with $\omega(\fx) = 1/2$. \textbf{Right}: Control sets for the curvature prior elastica model, embedded in the tangent space $\bE = \bR^2 \times \bR$.}
\label{fig_ControlSet} 
\end{figure*}

\subsection{Transformation of the Hamiltonian and Control Set}
In order to describe our variant of the classical elastica model with  curvature prior, we briefly adopt a more general point of view. A (sub-Finslerian) metric on the domain $\bM$ is a mapping $\cF:\bM \times \bE \to [0,\infty]$, obeying $\cF(\fx,\lambda \dfx) = \lambda\cF(\fx,\dfx)$ for any $\fx\in \bM$, $\dfx \in \bE$, $\lambda > 0$, satisfying $\cF(\fx,\mathbf{0})=0$, and such that the \emph{control sets} $\cB(\fx)$ defined by~\cref{eqdef:control_set} are non-empty, convex, compact, and depend continuously on the base point $\fx \in \bM$ w.r.t.\ the Hausdorff distance~\citep{chen2017global}. Some control sets obeying these properties are illustrated in the right column of Fig.~\ref{fig_ControlSet}, where the corresponding~\emph{Hamiltonian} $\cH$ is defined by Eq.~\eqref{eqdef:hamiltonian}.

\begin{proposition}
\label{prop:metric_linear_composition}
Let $\cF$ be a metric on the domain $\bM$, let $\cL_\fx : \bE \to \bE$ be a linear invertible map depending continuously on the parameter $\fx\in \bM$, and let $\kF(\fx,\dfx) := \cF(\fx,\cL_\fx \dfx)$ for all $\dfx \in \bE$. Then $\kF$ is a metric on $\bM$, and the corresponding control sets $\kB$
and Hamiltonian $\kH$ are obtained from those  of the metric $\cF$ denoted $\cB$ and $\cH$ as follows
\begin{align*}
	\kB(\fx) &= \cL_\fx^{-1}\cB(\fx), &
	\kH(\fx,\hfx) &= \cH(\fx,\cL_\fx^{-\top} \hfx).
\end{align*}
\end{proposition}

\begin{proof}
	The expression of the transformed control set $\kB(\fx)\subset \bE$ follows from the series of equivalences: 
\begin{equation*}
\dfx \in \kB(\fx) \Leftrightarrow \cF(\fx,\cL_\fx \dfx) \leq 1 \Leftrightarrow \cL_\fx \dfx \in \cB(\fx) \Leftrightarrow \dfx \in \cL_\fx^{-1} \cB(\fx).	
\end{equation*}	
By linearity of $\cL_\fx^{-1}$, the set $\kB(\fx)$ is like $\cB(\fx)$ non-empty, convex, and compact. In addition, the set $\kB(\fx)$ depends continuously on the point $\fx\in \bM$, by continuity of $\cB(\fx)$ and $\cL_\fx^{-1}$. Noting that $\kF(\fx,\lambda \dfx) = \cF(\fx,\lambda \cL_\fx \dfx) = \lambda \cF(\fx,\cL_\fx \dfx) = \lambda \kF(\fx,\dfx)$ and that $\kF(\fx,\mathbf{0}) = \cF(\fx,\mathbf{0}) = 0$, by linearity of $\cL_\fx$, we obtain as announced that $\kF$ is a metric. Finally, the announced expression of the Hamiltonian $\kH$ follows from the computation 
\begin{equation*}
\max_{\dfx \in \kB(\fx)} \<\hfx,\dfx\>
= \max_{\dfx \in \cB(\fx)} \left\langle\hfx,\cL_\fx^{-1} \dfx\right\rangle
= \max_{\dfx \in \cB(\fx)} \left\langle\cL_\fx^{-\top}\hfx, \dfx\right\rangle,
\end{equation*}
which concludes the proof.	
\qed
\end{proof}

\begin{proposition}	
\label{prop_ExplicitHamiltonianPrior}
The considered variant of the classical elastica model, equipped with a continuous curvature prior map $\omega : \bM \to \bR$, defines a metric with explicit expression in~\cref{eq_PriorMetric}, with control sets 	
\begin{align}
\nonumber
\kB(\fx):= \Big\{& \dfx = (\dx,\dtheta) \in \bE \mid \dx = \dn(\theta)\|\dx\|,\\
\label{eq:elastica_prior_control_set}
& \left(\|\dx\|-\tfrac{1}{2}\right)^2 + \beta^2 \left(\dtheta - \omega(\fx)\|\dx\|\right)^2\leq\tfrac{1}{4}\Big\}, 
\end{align}
for all $\fx = (x,\theta) \in \bM$, and whose Hamiltonian reads 
\begin{align*}
\kH&(\fx,\hfx)=\\
&\frac{1}{8}\left(\<\hfx,\dfn(\theta,\curPr(\fx))\>+\sqrt{\<\hfx,\dfn(\theta,\curPr(\fx))\>^2+(\htheta/\beta)^2)}\right)^2,
\end{align*}
for all co-vectors $\hfx = (\hfx,\htheta) \in \bE^*$. Recall that $\dfn(\theta,a) := ((\cos\theta,\sin \theta),a)\in \bE := \bR^2 \times \bR$. This Hamiltonian can also be written as: 
\begin{align}
\label{eq_ExplicitHamiltonianPrior2}
&\kH(\fx,\hfx)=\frac{3}{8}\int^{\pi/2}_{-\pi/2}\<\hfx,\dkq(\fx,\varphi)\>_+^2\cos\varphi d\varphi,
\end{align}
where  
$\dkq(\fx,\varphi)\in\bE$ is defined for any $\vp \in [-\pi/2,\pi/2]$ as 
\begin{equation*}
\dkq(\fx,\varphi):=\left(\dn(\theta)\cos\varphi,\curPr(\fx)\cos\varphi+\beta^{-1}\sin\varphi\right).
\end{equation*}
\end{proposition}

\begin{proof}
Let us fix $\fx = (x,\theta) \in \bM$ in the following. Consider the linear map $\cL_\fx : \bE \to \bE$ defined for all $\dfx = (\dx,\dtheta) \in \bE$ as 
\begin{equation*}
	\cL_\fx(\dfx) := (\dx,\, \dot \theta - \omega(\fx) \<\dn(\theta),\dx\>).
\end{equation*}
For clarity, let us also provide the matrix of this operator in the canonical basis of the three-dimensional space $\bE = \bR^2 \times \bR$
\begin{equation}
	\label{eq:Lx_matrix}
	[\cL_\fx] = 
	\begin{pmatrix}
		1 & 0 & -\omega(\fx)\cos \theta\\
		0 & 1 & -\omega(\fx)\sin \theta	\\
		0 & 0 & 1
	\end{pmatrix}.
\end{equation}
Then $\cL_\fx$ depends continuously on the point $\fx$, and is invertible in view of the triangular form of its matrix. Since $\kF(\fx,\dfx) = \cF(\fx,\cL_\fx \dfx)$, indeed compare~\cref{eq_PriorMetric} with~\cref{eq_ElasticaMetric},  the variant of the elastica model equipped with the curvature prior $\omega$ defines a valid metric whose control sets and Hamiltonian can be expressed in terms of those of the classical Euler-Mumford elastica model, see~\cref{prop:metric_linear_composition}. More precisely, noting  the equivalence $\dfx \in \kB(\fx) \Leftrightarrow \cL_\fx \dfx \in \cB(\fx)$ and using the characterization of $\cB(\fx)$ established in~\cref{eq:elastica_control_set}, we obtain~\cref{eq:elastica_prior_control_set}.
	From \cref{eq:Lx_matrix},  we can easily compute the inverse and the transposed inverse mappings, namely
\begin{align*}
\cL_\fx^{-1}(\dfx) &= (\dx,\, \dtheta + \omega(\fx) \<\dn(\theta),\dx\>), & 
\forall\dfx &= (\dx,\dtheta) \in \bE,\\
\cL_\fx^{-\top}(\hfx) &= (\hx + \omega(\fx) \dn(\theta),\, \htheta), &
\forall\hfx &= (\hx,\htheta) \in \bE^*.
\end{align*}
Note that $$\<\cL_\fx^{-\top}\hfx, \dfn(\theta,0)\> = \<\hfx, \cL_\fx^{-1}\dfn(\theta,0)\> = \<\hfx,\dfn(\theta,\omega(\fx))\>$$ and likewise that $$\<\cL_\fx^{-\top}\hfx,(\dn(\theta)\cos\varphi,\beta^{-1}\sin\varphi)\> = \<\hfx,\dkq(\fx,\varphi)\>.$$ 
Evaluating~\cref{eq_ElasticaHam1} and~\cref{eq_ElasticaHam2} at $\cL_\fx^{-\top}(\hfx)$, and using the previous identities, we obtain the announced expressions of the Hamiltonian $\kH$, which concludes the proof. 
\end{proof}

The integral form of the Hamiltonian $\kH$, formulated in~\cref{eq_ExplicitHamiltonianPrior2} can be approximated using the $L$-point Fejer quadrature rule for any positive integer $L$, following~\citep{mirebeau2018fast}. This yields
\begin{equation}
\label{eq_Approximation}
\kH(\fx,\hfx)=\sum_{l=1}^L \mu_l\<\hfx,\, \dkq(\fx,\varphi_l)\>_+^2 + \|\hfx\|^2 \cO( L^{-2}),
\end{equation}
where we denoted by $\mu_l>0$ the weights of the Fejer quadrature rule at the angles $\varphi_l =  (2 l-1-L)\pi/(2L)\in[-\pi/2,\pi/2]$, for any $1\leq l \leq L$. This approximation plays a crucial role in the HFM method, as presented in next section.

\begin{figure*}[t]
\setlength{\fboxsep}{0pt}
\centering
\subfigure[]{\includegraphics[width=0.235\linewidth]{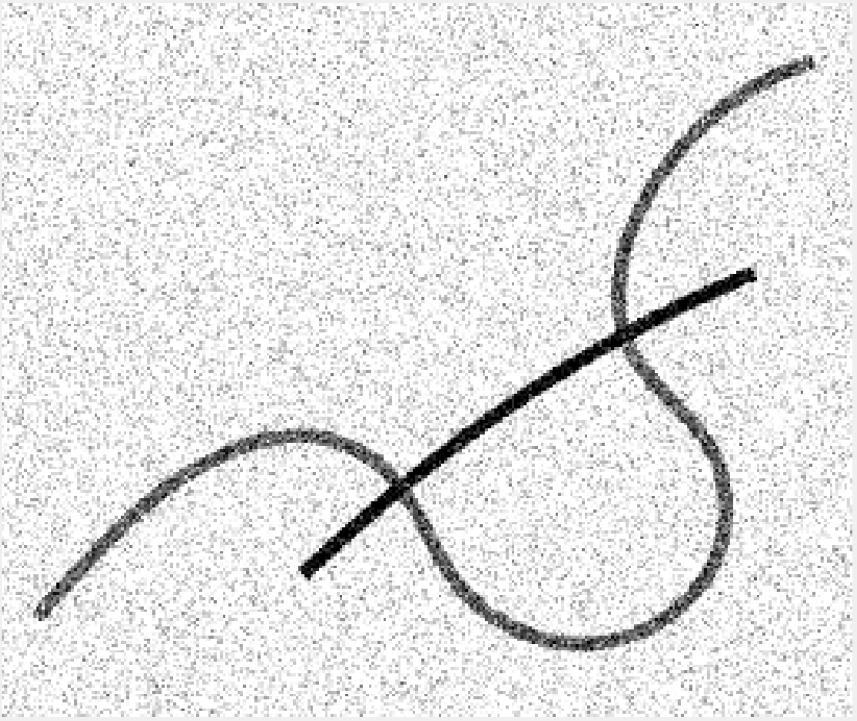}}~\subfigure[]{\includegraphics[width=0.235\linewidth]{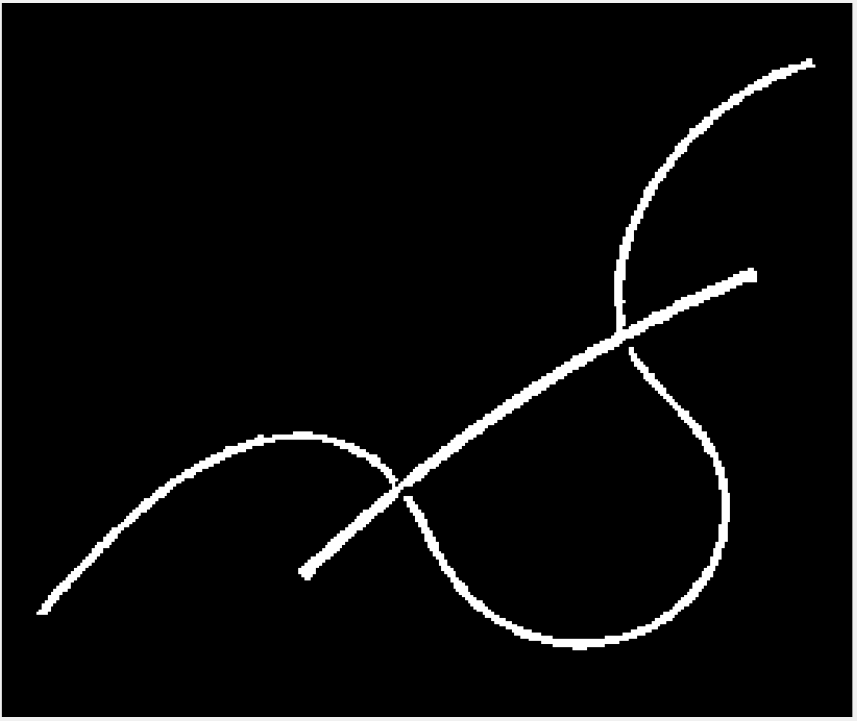}}~\subfigure[]{\fbox{\includegraphics[width=0.235\linewidth]{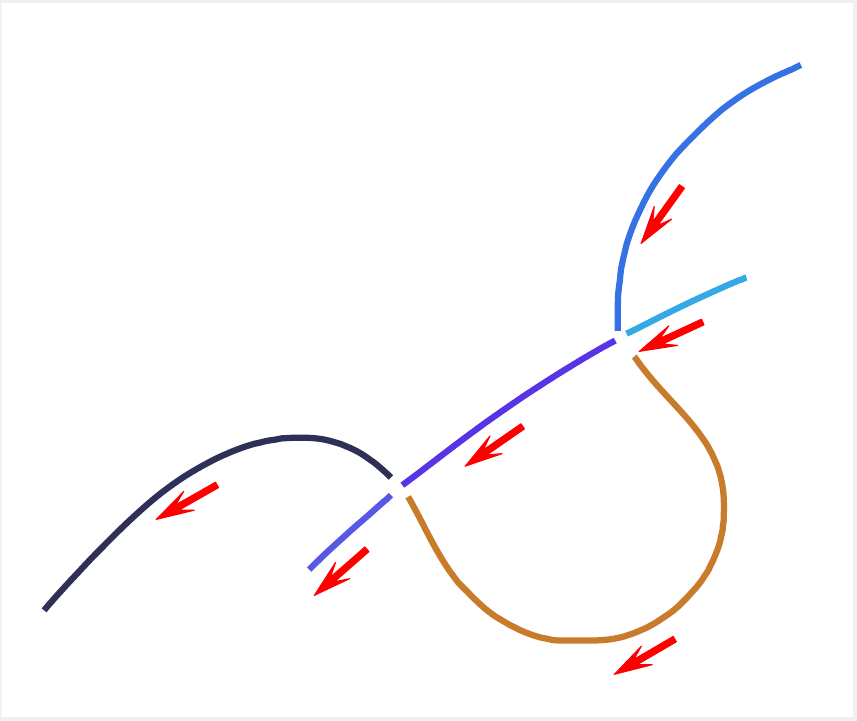}}}~\subfigure[]{\fbox{\includegraphics[width=0.235\linewidth]{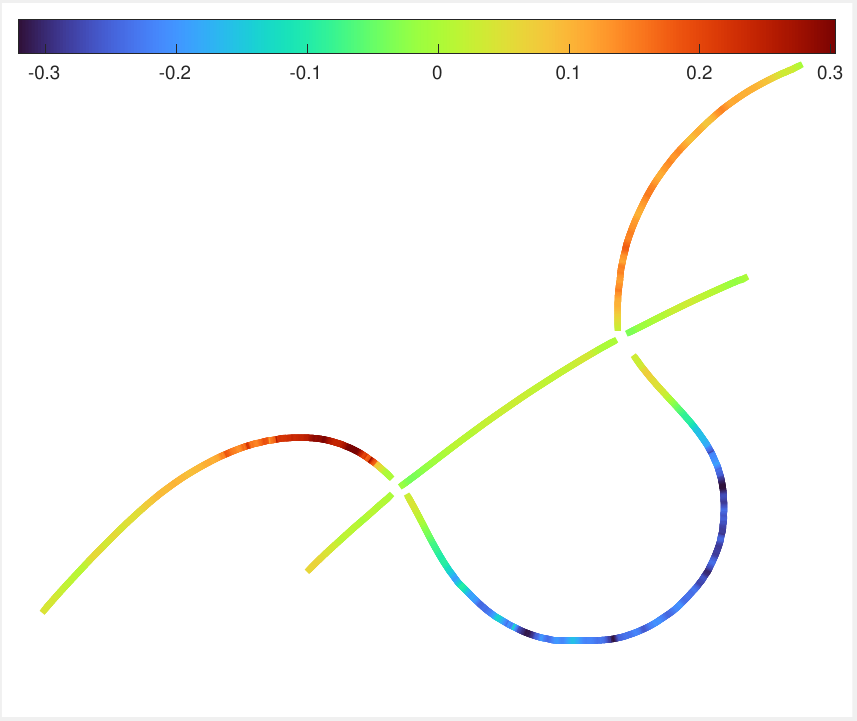}}}
\caption{An example for the construction of candidate centerlines. \textbf{a} The original image. 
\textbf{b} Visualizing the binary segmentation map $\Xi$. 
\textbf{c} The candidate centerlines $\rho_j$ indicated by different colors. The arrows illustrate the directions of the corresponding candidate centerlines. 
\textbf{d} Visualizing the curvature of the candidate centerlines.
}
\label{fig_step}
\end{figure*}

\section{Application to Tracking Curvilinear Structures}
\label{sec_App}
Geodesic paths are considered as a powerful tool for interactive curvilinear structure tracking~\citep{cohen1997global,li2007vessels,pechaud2009extraction,liao2018progressive,chen2019minimal}.  Following the spirit of this line of works, we propose a fast and reliable method for computing a curvature prior map $\omega$ from image data, allowing the curvature prior  elastica model to accurately extract curvilinear structures of interest. This computation method,  presented in the remainder of this section, consists of two main steps: (I) the generation of a set of disjoint curvilinear structure centerlines, and (II) the estimation of the tangent direction and curvature associated with those disjoint centerlines.

\subsection{Generation of the Disjoint Centerlines} 
Centerlines serve as an effective descriptor for curvilinear structures, by which geometric features such as their tangent direction and curvature properties can be estimated. A broad variety of approaches have been introduced to handle the task of extracting centerlines from curvilinear structure networks which usually include complex junctions, e.g.~\citep{kaul2012detecting,cetin2012vessel,turetken2016reconstructing}. Among those methods, the skeletonization schemes, adopted in this work,  provide a powerful way for producing centerlines as the skeletons of the pre-segmented curvilinear structures~\citep{kimmel1995skeletonization,siddiqi2002hamilton,lam1992thinning}. 

Let $\Xi:\Omega\to\{0,1\}$ represent the binary segmentation of curvilinear structures contained in an image, where $\Xi(x)=1$ implies that the position $x$ belongs to a curvilinear structure, while $\Xi(x)=0$ indicates a background position $x$. We compute the skeleton or medial axis of the elongated shapes $\{x\in\Omega\mid\Xi(x)=1\}$, and remove all the junction points of these skeleton structures, yielding $J$ disjoint skeleton segments, where $J$ is a positive integer. After applying a curve smoothing procedure to regularize these skeleton segments, we obtain a family of smooth curves $\rho_j:[0,1]\to\Omega$, indexed by $1\leq j\leq J$, as the parametrization of the regularized skeleton segments. In what follows we refer to these curves $\rho_j$ as \emph{candidate centerlines}. 

In practice, we rely on morphological filters~\citep{lam1992thinning}, which benefit from a low computational complexity, to generate skeletons from pre-segmentation data.
Fig.~\ref{fig_step} illustrates an example for constructing the candidate centerlines $\{\rho_j\}_{1\leq j\leq J}$ using a synthetic image. Fig.~\ref{fig_step}b visualizes the corresponding binary segmentation map $\Xi$. In Fig.~\ref{fig_step}c, the candidate centerlines are illustrated using different colors, and arrows indicate their directions. The curvature of these candidate centerlines is depicted in Fig.~\ref{fig_step}d.

Note that a curvilinear structure of interest usually corresponds to the concatenation of several candidate centerlines, together with the crucial curves  connecting pairs of adjacent centerlines. For this reason, the skeletonization process is not an appropriate solution to the curvilinear structure tracking problem considered, but rather an efficient descriptor for extracting the geometric 
features of parts of the curvilinear structures.

\subsection{Construction of the Curvature Prior Map $\omega$}
We propose to estimate the curvature prior map $\omega$ using the computed candidate centerlines $\rho_j$, whose curvature and unit normal are respectively denoted by $\kappa_j:[0,1]\to]-\infty,\infty[$ and $\cN_j:[0,1]\to\bR^2$, for $1\leq j \leq J$.

Let $0<U<\min_j\{1/\|\kappa_j\|_\infty\}$ be a bounded constant. Consider the mapping $\Psi_j:[0,1]\times\,[-U,U]\,\to\bR^2$, depending on the curve parameter and the deviation from $\rho_j$, defined as 
\begin{equation*}
\Psi_j(u,\lambda):=\rho_j(u)+\lambda\cN_j(u)
\end{equation*}
for any $u\in[0,1]$ and $\lambda\in[-U,U]$. The bounding constraint on the parameter $U$ implies that the mappings $\Psi_j$ are injective, and  each of them  parametrizes a tubular neighborhood $\kT_j\subset\Omega$ of width $U$ surrounding the candidate centerline $\rho_j$
\begin{equation*}
\kT_j:= \left\{x=\Psi_j(u,\lambda)\in\Omega\mid\forall u\in[0,1]\text{~and~}\forall\lambda\in[-U,U]\right\}.	
\end{equation*}
However, such a tubular neighborhood $\kT_i$ may intersect with a distinct one $\kT_j$, i.e.\ $\kT_i\cap\kT_j \neq \emptyset$, when their centerlines $\rho_i$ and $\rho_j$ are close to each other. In order to tackle this problem, we partition the domain $\Omega$ into a series of disjoint regions $\cV_j\subset\Omega$ for $1\leq j \leq J$
\begin{equation}
\label{eq_Voronoi}
\cV_j:=\left\{x\in\Omega~|~d_{\rho_j}(x)<d_{\rho_i}(x),~\forall i\neq j\right\},
\end{equation}
where $d_{\rho_j}(x)$ denotes the Euclidean distance between $x$ and the candidate centerline $\rho_j$. These regions are also known as the Voronoi cells. Accordingly, for each $\rho_j$ a modified tubular neighborhood $\cT_j\subset\Omega$ can be constructed  by
\begin{equation}
\label{eq_FinalTubeNeigh}
\cT_j:=\kT_j\cap\cV_j,	
\end{equation}
which satisfies $\cT_j\cap \cT_i=\emptyset$ for any $i\neq j$. In practice, the construction of the Voronoi cells $\cV_j$ can be efficiently implemented by propagating an unsigned Euclidean distance map simultaneously from all the candidate centerlines~\cite{saye2011voronoi,saye2012analysis}.

We consider two scalar-valued functions $\phi:\cup_j\cT_j\to\bR$ and $\vartheta:\cup_j\cT_j\to\bS^1$ using all the candidate centerlines $\rho_j$, $1 \leq j \leq J$, such that for any physical position $x=\Psi_j(u,\lambda)\in\cT_j$, one respectively has 
\begin{align*}
\phi(x):=\kappa_j(u),~\text{ and ~}\dn(\vartheta(x))\|\rho^\prime_j(u)\|=\rho^\prime_j(u).
\end{align*}
The function $\phi$ extends the curvature of the candidate centerlines $\rho_j$ to their respective neighborhood regions $\cT_j$ of width $U$. Furthermore, for a physical position $x\in\cT_j$, the angles $\vartheta(x)$ and $\vartheta(x)+\pi$ correspond to the two possible orientations that the elongated structure should have at the position $x$.
With these definitions, the curvature prior map $\omega$ can be generated in terms of the estimated functions $\phi$ and $\vartheta$. More precisely,  we define the curvature prior map $\omega$ for any point $\fx=(x,\theta)\in\bM$ as follows
\begin{equation}
\label{eq_PriorConstruction}
\omega(\fx):=
\begin{cases}
\phi(x)\sign\big(\langle\dn(\theta),\dn(\vartheta(x))\rangle\big),&\forall\fx\in\cup_j\cT_j\times\bS^1,\\
0,	&\text{otherwise},
\end{cases}
\end{equation} 
where $\sign(a)$ stands for the sign  of a scalar value $a\in\bR$.
In the first line of~\cref{eq_PriorConstruction},  for any  physical position $x\in\cup_j\cT_j$ we set $\omega(x,\theta)=\omega(x,\vartheta(x))=\phi(x)$ if  the angular position $\theta$ is closer to $\vartheta(x)$ in the sense of the Euclidean scalar product of the vectors $\dn(\theta)$  and $\dn(\vartheta(x))$, and  $\omega(x,\theta)=\omega(x,\vartheta(x)+\pi)=-\phi(x)$ if $\theta$ is closer to $\vartheta(x)+\pi$. 

\subsection{Computation of the Cost Function $\psi$} 
The multi-orientation data-driven cost function $\psi(x,\theta)$ should have low values when both (i) the physical position $x$ is close to the centerline of a curvilinear structure, and (ii) the direction $\dn(\theta)$ closely aligns with the tangent of this centerline at the physical position $x$. For that purpose and following~\citep{chen2017global,duits2018optimal,bekkers2015pde}, we build the cost function $\psi$ for any point $\fx = (x,\theta) \in\bM$ as follows
\begin{equation}
\label{eq_CostFunction}
\psi(\fx)=\exp\left(-\alpha\,\frac{g(\fx)}{\|g\|_\infty}\right),
\end{equation}
where $\alpha>0$ is a positive weighting parameter, and where $g:\bM\to [0,\infty[$  is an orientation score map usually taken as an efficient tool in image analysis~\citep{franken2009crossing,hummel1983foundations,bekkers2014multi,parent1989trace}. This map carries the curvilinear appearance and anisotropy features, and obeys $g(x,\theta)=g(x,\theta+\pi)$, $\forall(x,\theta)\in\bM$. In this work, the steerable optimally oriented flux filter~\citep{law2008three} is chosen as the curvilinear feature extractor in order to compute the orientation score $g$. We refer to~\citep{chen2017global} for more details on the computation of the orientation score map $g$ from image data involving curvilinear structures.

\begin{figure}[t]
\centering
\includegraphics[width=0.75\linewidth]{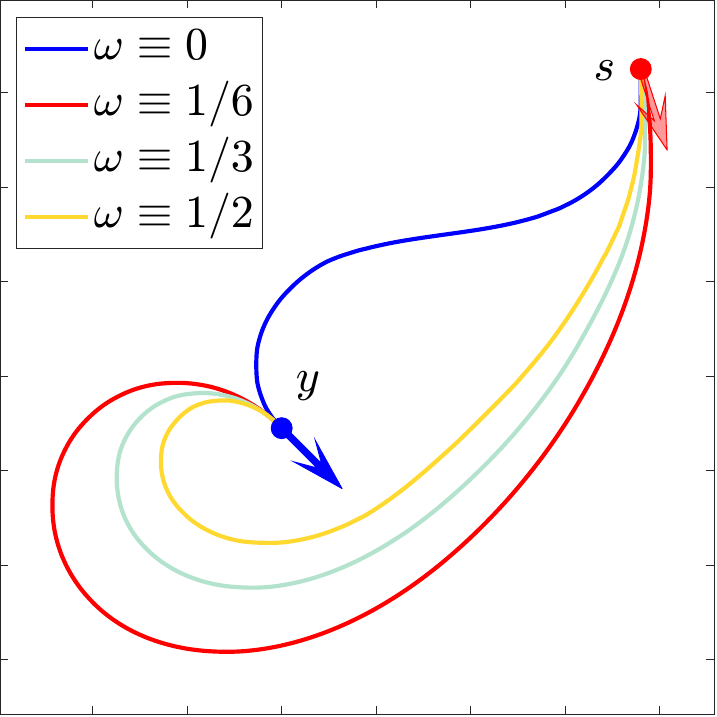}
\caption{Illustration of the physical projections of geodesic paths for the curvature prior elastica model with different values of $\omega$.  The red and blue dots with the associated arrows represent the source and target points of the geodesic paths, respectively. }
\label{fig_ConstVelocity}
\end{figure}

\subsection{Initialization and Geodesic Paths Tracking}
The curvilinear structure tracking method can be implemented in an interactive manner, requiring a source point $\fs=(s,\theta_s)\in\bM$ and a target point $\fy=(y,\theta_y)\in\bM$ as initialization. In our model, we assume that the physical positions $s$ and $y$ are provided by the user. Furthermore,  the corresponding angles $\theta_s$ and $\theta_y$ can be either given by the user in complex scenarios, or be estimated from the cost function $\psi$ as 
\begin{equation}
\theta_s = \underset{\theta\in[0,\pi[}{\arg\min}~\psi(s,\theta).
\end{equation}

In addition to the points $\fs$ and $\fy$ as described above, we generate another pair of source and target points $\tilde\fs=\big(s,\theta_s+\pi\big)$ and $\tilde\fy=\big(y,\theta_y+\pi\big)$. In other words, $\fs$ and $\tilde\fs$ (resp. $\fy$ and $\tilde\fy$) share the same physical position but different angular positions.  In this case, the goal is to find a minimal geodesic path from the source point set $\ks=\{\fs,\tilde\fs\}$ to the target point set $\ky=\{\fy,\tilde\fy\}$,  as considered in~\citep{chen2017global}. For that purpose, a geodesic distance map $\cD_\ks$ is estimated by solving the HJB PDE 
\begin{equation*}
\kH(\fx,d\cD_\ks(\fx))=\frac{1}{2}\psi(\fx)^2,
\end{equation*}
for all $\fx \in \bM \sm \ks$,
with the boundary condition $\cD_{\ks}(\fx)=0$ for any point $\fx\in\ks$, which is a straightforward generalization of the single source case described in \cref{eq_HJB_biased}. 
The geodesic of interest is backtracked, as described in~\cref{eq_ODE}, from the point of the target set $\ky$ whose geodesic distance value is smallest, and by the nature of the viscosity solution the backtracking procedure ends at the point of the seed set $\ks$ which is closest w.r.t.\ the metric. For computational efficiency, the procedure for estimating the respective geodesic distances can be terminated immediately once either point of the target set $\ky$ is reached by the fast marching front, and this optimization is valid since the fast marching front advances monotonically. Note that this geodesic path tracking procedure allows the user \emph{not} to distinguish between the source and target points, which is convenient in practice.

\begin{figure}[t]
\centering
\subfigure[]{\includegraphics[width=0.46\linewidth]{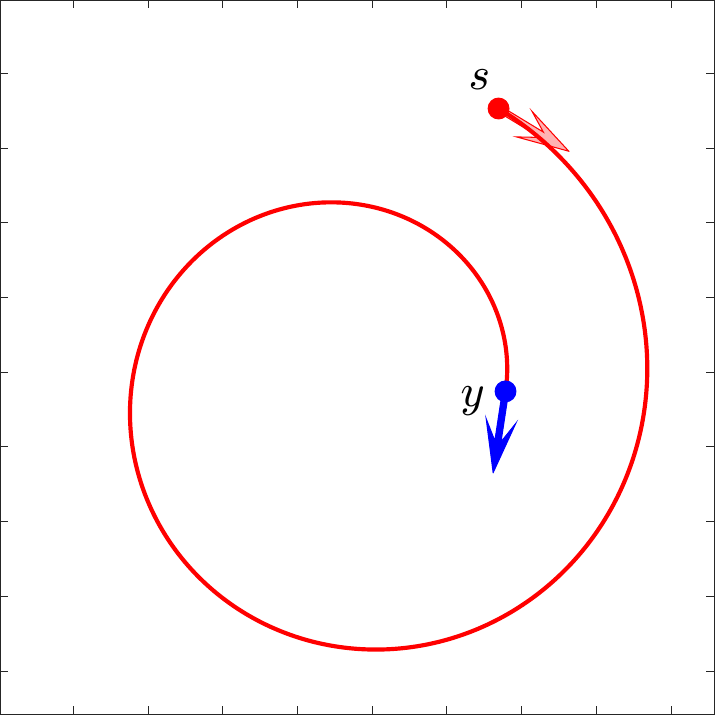}}~\subfigure[]{\includegraphics[width=0.46\linewidth]{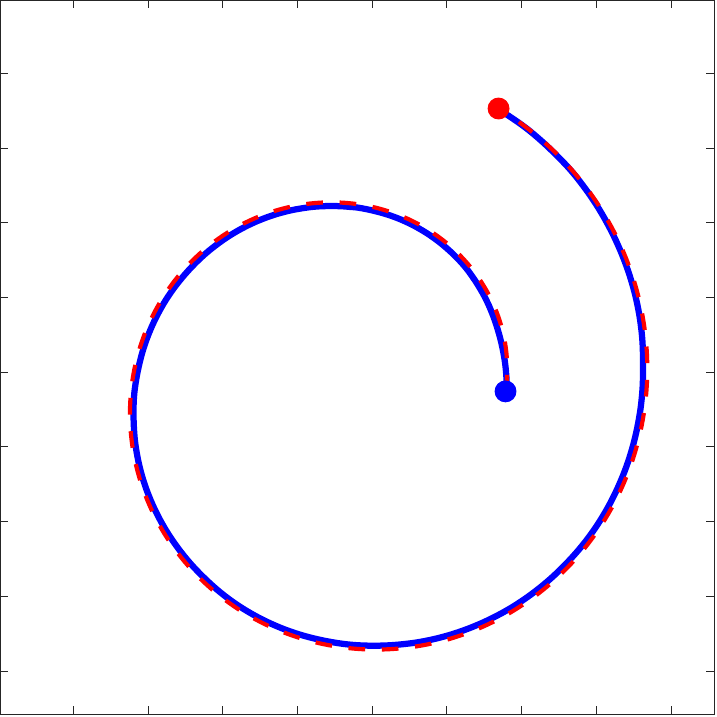}}
\caption{Illustration of the extraction of a smooth curve using the curvature prior elastica model.  \textbf{a} A smooth planar curve whose endpoints $s$ and $y$ are respectively denoted by the red and blue dots. The arrows represent the tangent vectors at the endpoints $s$ and $y$ of the planar curve. \textbf{b} The blue solid line denotes the physical projection of the  geodesic path from the curvature prior elastica model, and the red dashed line illustrates the given smooth curve.}
\label{fig_CurveFitting}
\end{figure}

\section{Experimental Results}
\label{sec_Exp}
In this section, we present numerical results obtained with the introduced curvature prior  elastica model. Our experiments are meant to illustrate the influence of the curvature prior term $\omega$, and for that purpose we perform qualitative and quantitative comparisons with the classical elastica model~\citep{chen2017global,mirebeau2018fast}, on curvilinear structure tracking tasks applied to both synthetic and real images.
In the remainder of this section, both geodesic models are configured using the identical setting, except for the curvature prior term which is set to $\omega\equiv0$ for the classical elastica model.

Fig.~\ref{fig_ConstVelocity} illustrates several geodesic paths from the curvature prior elastica geodesic model where the term $\omega$ is set as a \emph{constant function} with different values. The cost function is set as $\psi\equiv 1$.  The orientation-lifted source point $\fs = (s,\theta_s) \in \bM$ is shown as a red dot at the position $s$ equipped with an arrow pointing in the direction $\dn(\theta_s)$, and likewise for the target point $\fy=(y,\theta_y)$ in blue. 
In this figure, the planar physical projections of four distinct geodesic paths corresponding to different values of $\omega$ are illustrated. The blue line corresponds to the geodesic path from the classical elastica model (i.e.\ $\omega\equiv0$).  As the value of $\omega$ increases, we observe in this figure that the portions of these geodesic paths surrounding the target point feature turns with diminishing turning radius.  

\begin{figure}[t]
\centering
\includegraphics[width=0.31\linewidth]{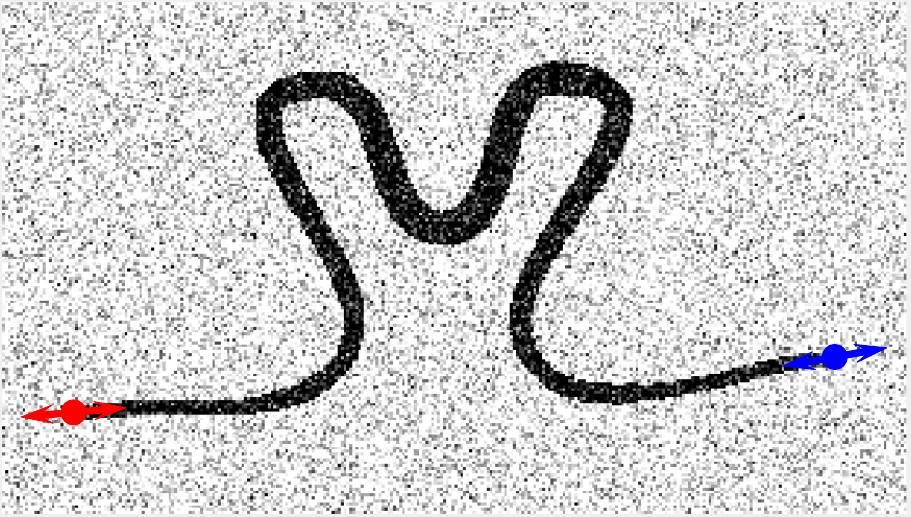}\,\includegraphics[width=0.31\linewidth]{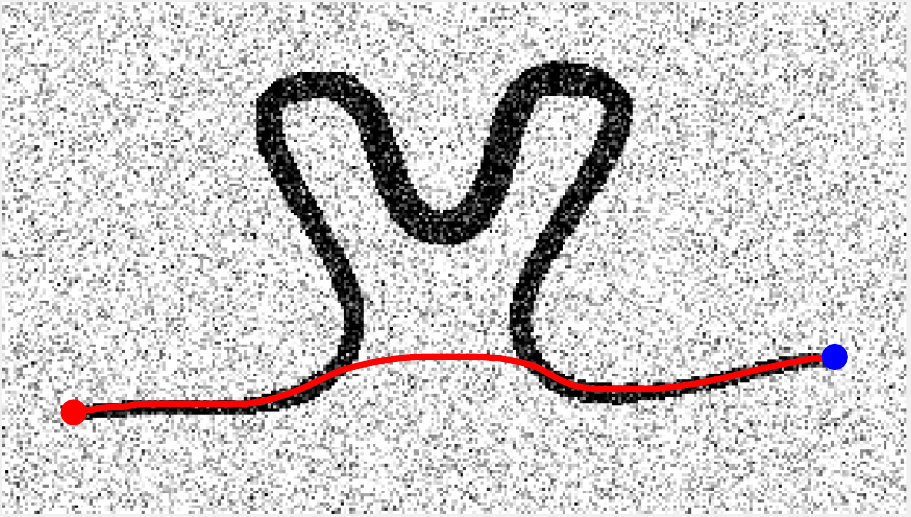}\,\includegraphics[width=0.31\linewidth]{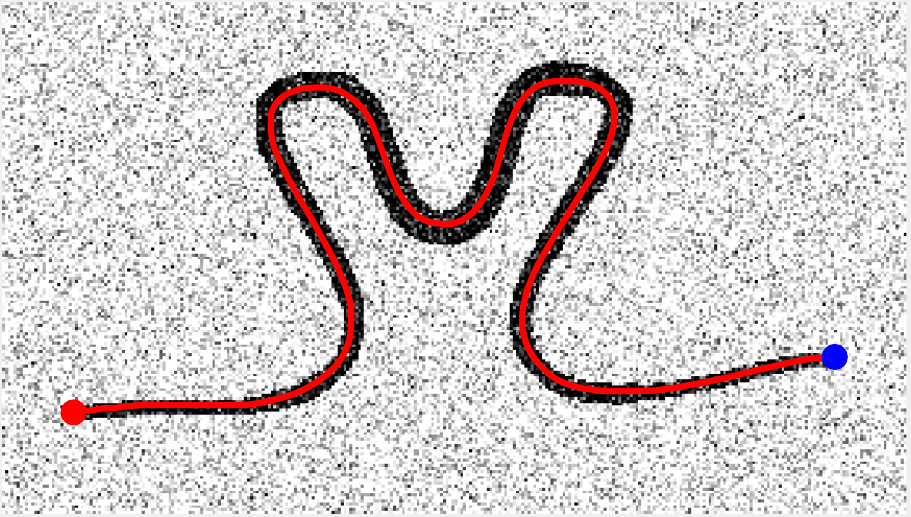}\\
\includegraphics[width=0.31\linewidth]{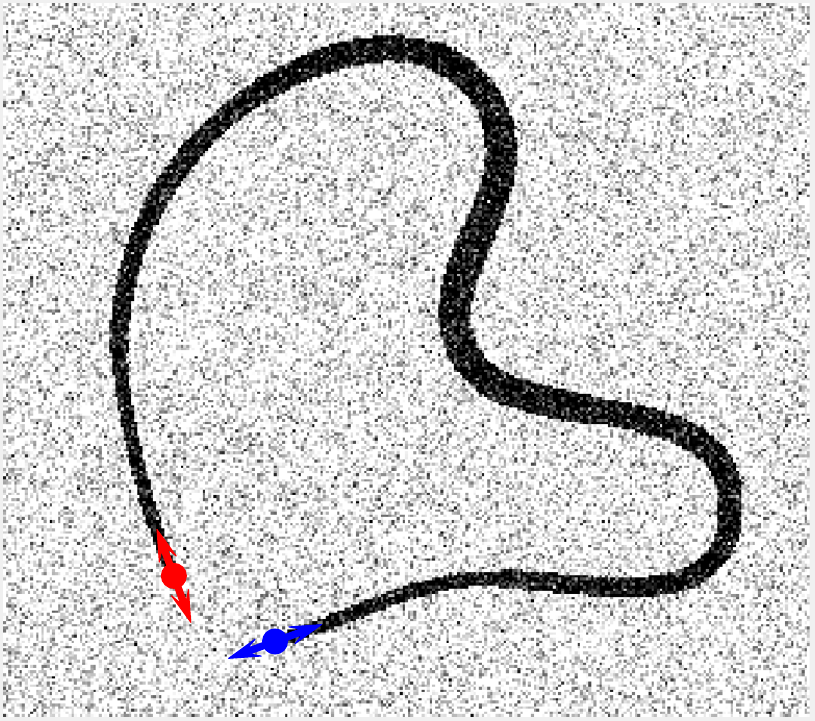}\,\includegraphics[width=0.31\linewidth]{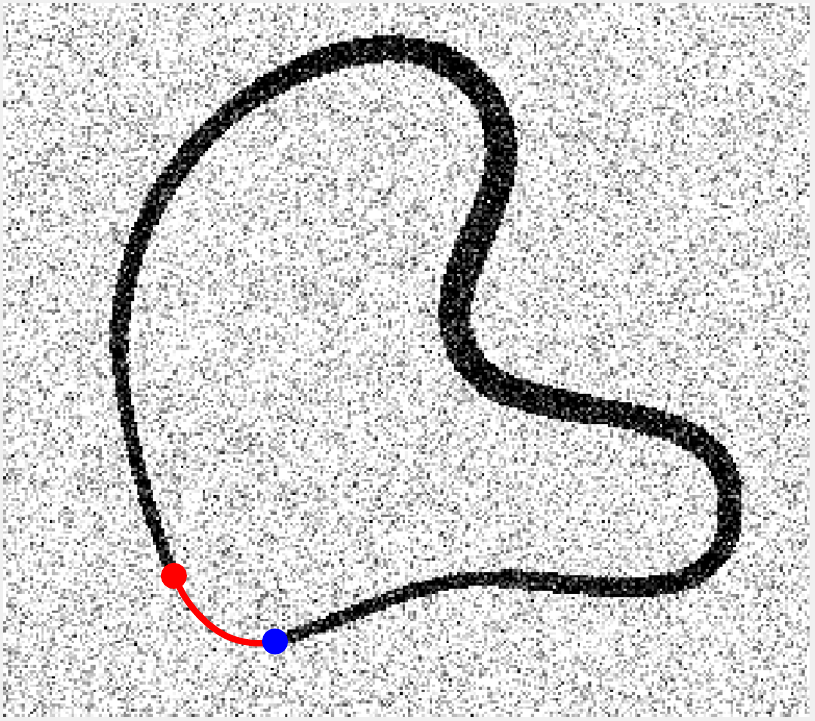}\,\includegraphics[width=0.31\linewidth]{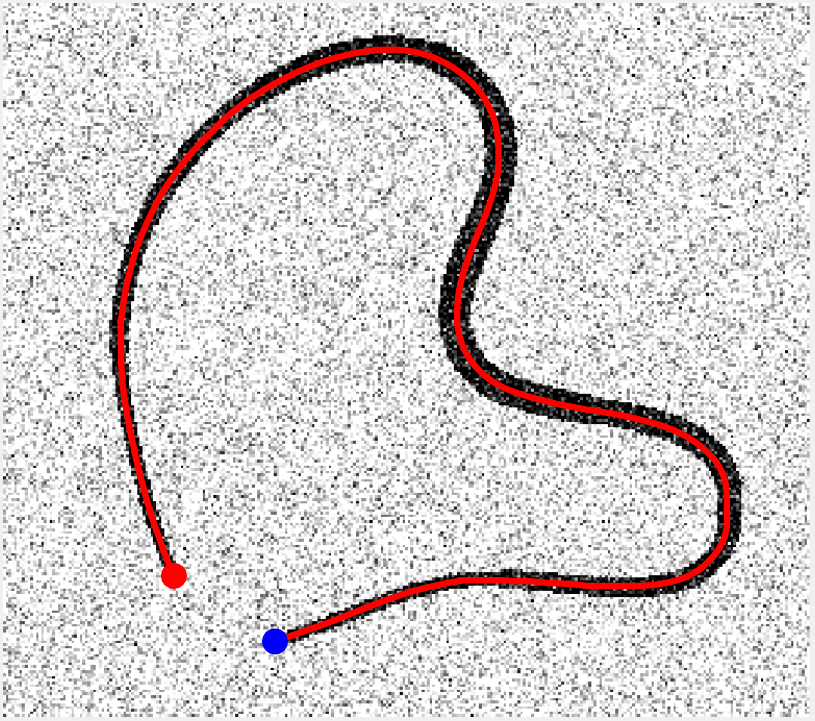}\\
\includegraphics[width=0.31\linewidth]{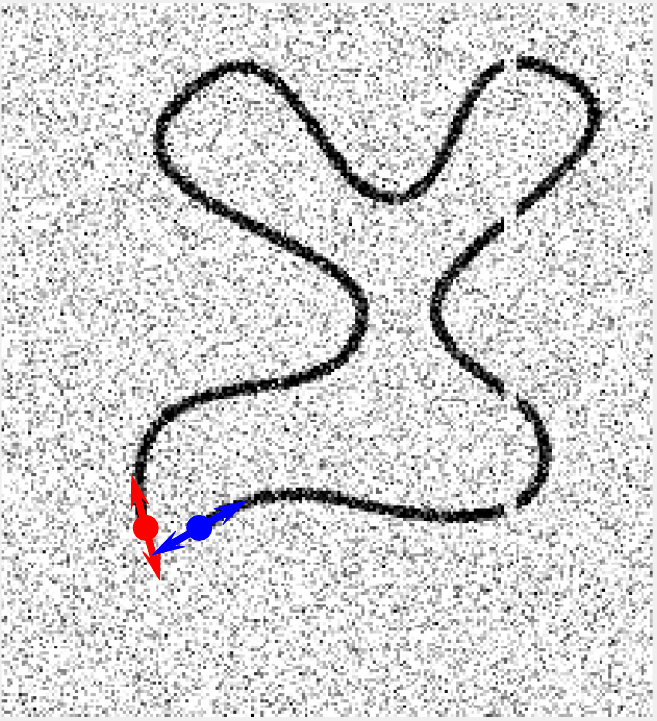}\,\includegraphics[width=0.31\linewidth]{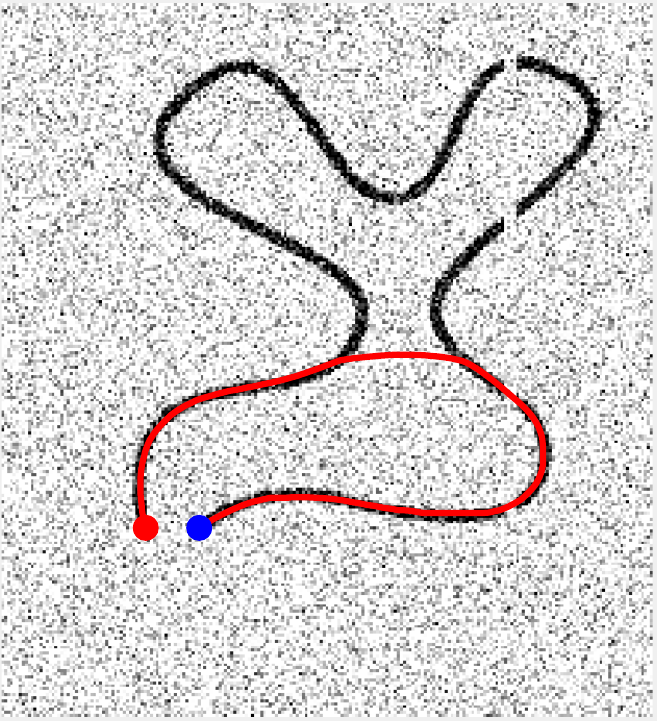}\,\includegraphics[width=0.31\linewidth]{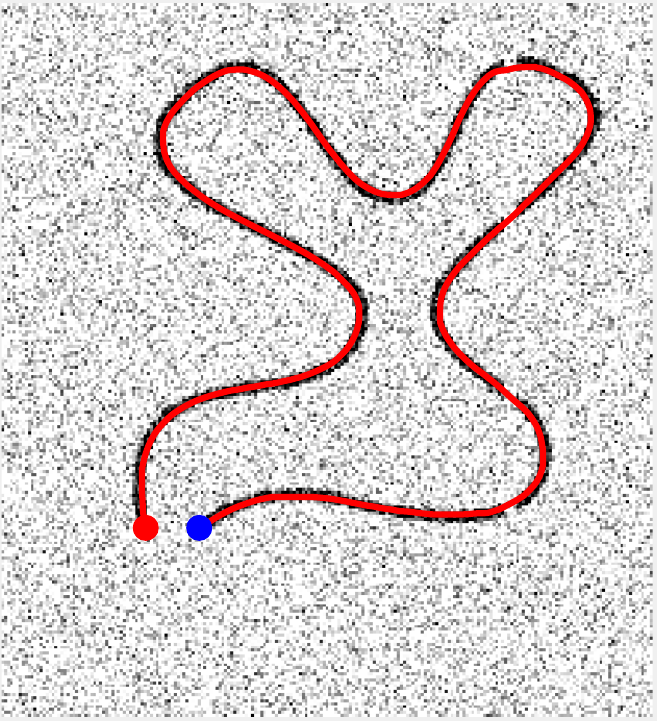}\\
\includegraphics[width=0.31\linewidth]{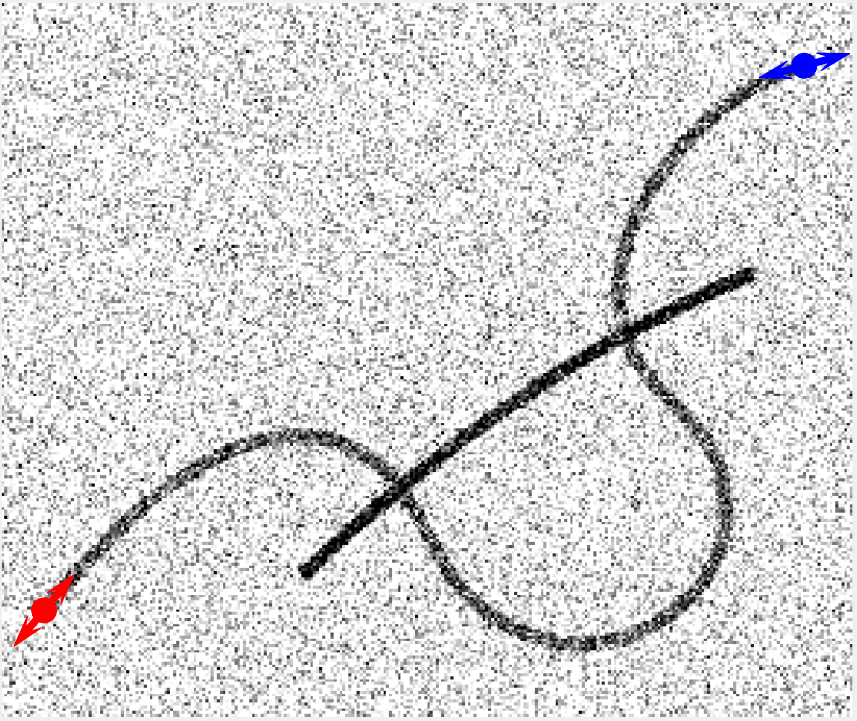}\,\includegraphics[width=0.31\linewidth]{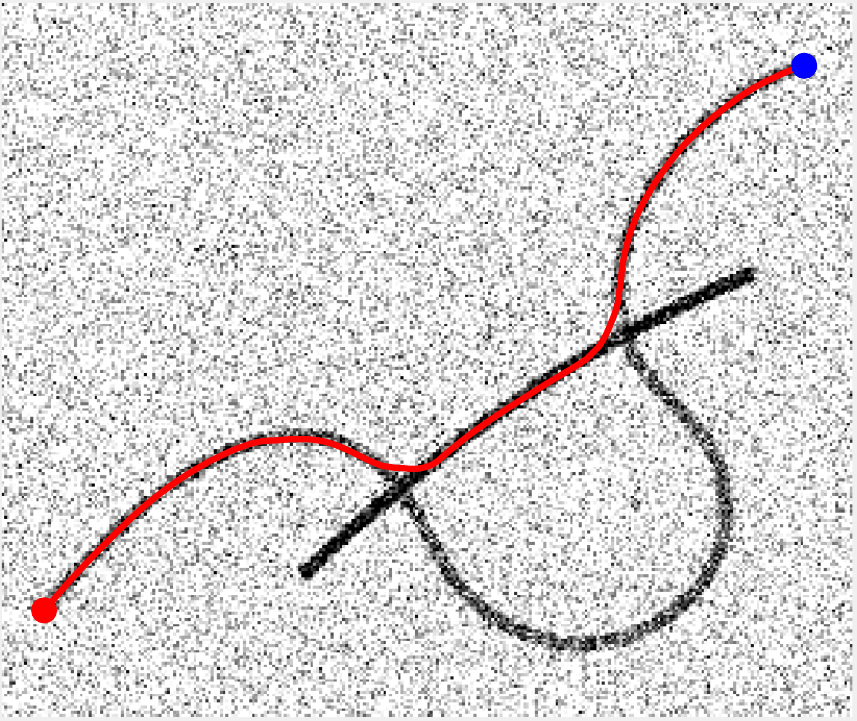}\,\includegraphics[width=0.31\linewidth]{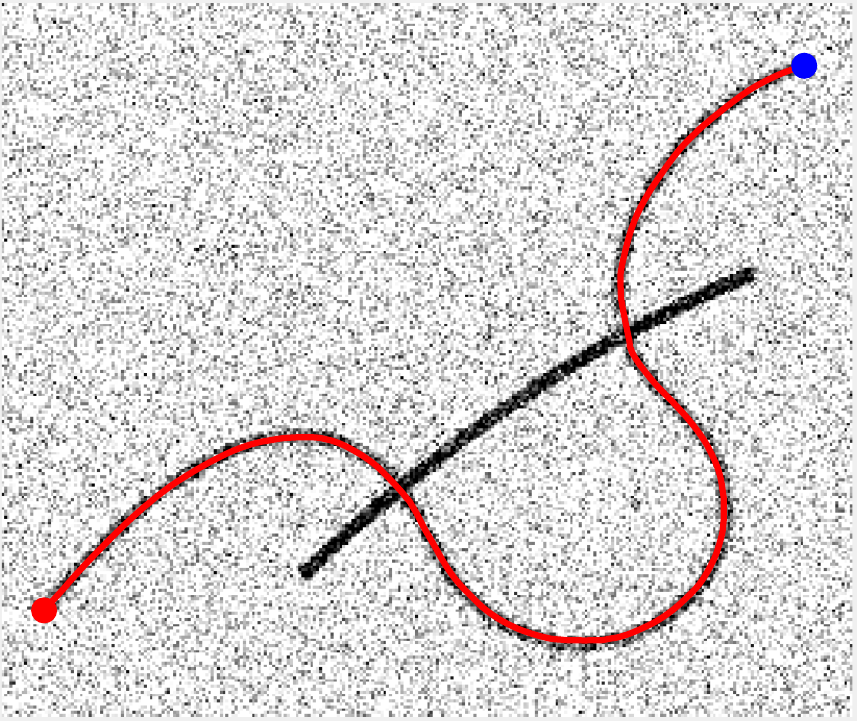}\\
\includegraphics[width=0.31\linewidth]{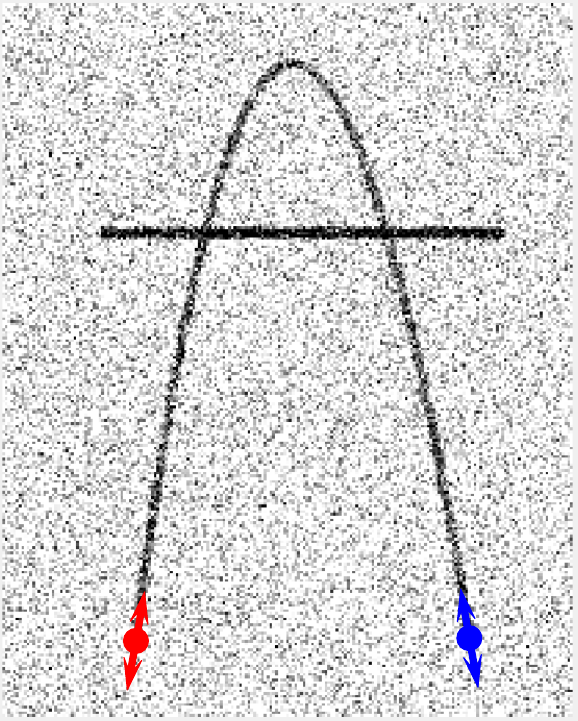}\,\includegraphics[width=0.31\linewidth]{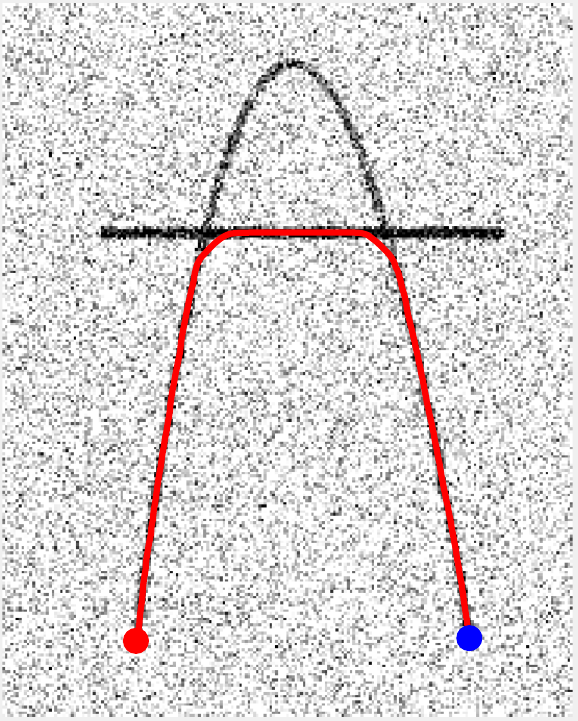}\,\includegraphics[width=0.31\linewidth]{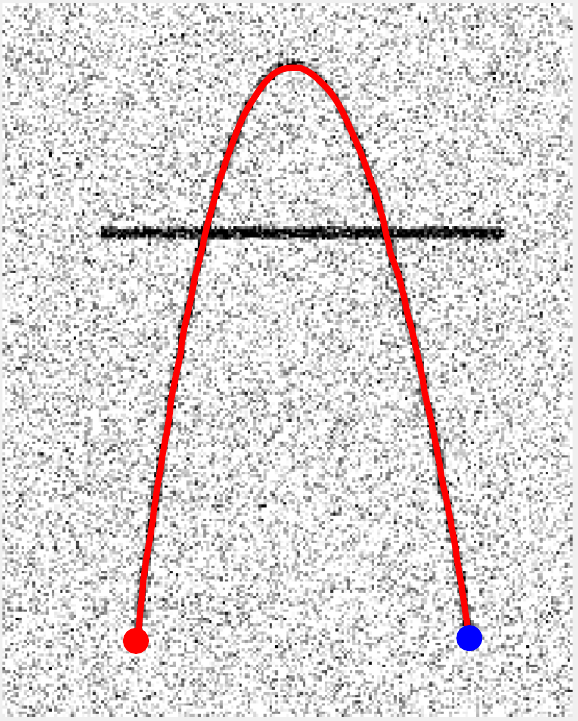}
\caption{Qualitative comparison results on synthetic images  blurred by additive Gaussian noise.  The initialization for each image is shown in column $1$, where the red and blue dots with the respective arrows indicate the source and target points. The physical projections of the geodesic paths from the classical elastica model and the curvature prior elastica model are respectively illustrated with red lines in columns $2$ and $3$.}
\label{fig_CompSyn} 
\end{figure}

Fig.~\ref{fig_CurveFitting} illustrates a geodesic path computed using the curvature prior elastica model, where the curvature prior term $\omega : \bM \to \bR$ is estimated pointwise from a given smooth curve as described in~\cref{eq_PriorConstruction}, and the cost function $\psi \equiv 1$ is constant. The given reference curve is shown on Fig.~\ref{fig_CurveFitting}a, and connects an endpoint $s$ to another one $y$, with tangent directions $\dn(\theta_s)$ and $\dn(\theta_y)$ shown as red and blue arrows respectively. The blue solid line in Fig.~\ref{fig_CurveFitting}b illustrates the physical projection of the  geodesic path obtained using the variant of the  elastica model, equipped with the estimated curvature prior map $\omega$ described above, from the source point $\fs=(s,\theta_s)$ to the target point $\fy=(y,\theta_y)$. In Fig.~\ref{fig_CurveFitting}b, the reference smooth curve is shown in red and dashed. When the curvature penalization parameter $\beta$ is sufficiently large, as in this experiment where we use $\beta=20$, the physical projection of the computed geodesic path should be almost overlaid onto the reference smooth curve, since they have almost identical curvature (recall that by construction the energy functional penalizes the square of their curvature difference), and since they share their endpoints and tangents at these endpoints. More precisely, the physical projection of the computed geodesic path converges to the reference curve as $\beta \to \infty$, and to a straight line joining the endpoints as $\beta \to 0$.
The computed geodesic path shown in Fig.~\ref{fig_CurveFitting}b indeed well approximates the given reference curve, illustrating the effectiveness of the curvature prior in the considered variant of the classical elastica model.

\subsection{Qualitative Comparison}
In Fig.~\ref{fig_CompSyn}, we compare results obtained with the classical elastica model, and with the curvature prior elastica model, on five synthetic images.
A common feature for the curvilinear structures in those images is that each of them involves strongly bent segments. 
The first column of Fig.~\ref{fig_CompSyn} shows the initialization information for each synthetic image, where the blue and red dots together with the arrows stand for the source and target points respectively used to define the sets $\ks$ and $\ky$. Columns $2$ and $3$ illustrate the physical projections of the computed geodesic paths from the classical elastica model and the curvature prior elastica model, respectively. 
More specifically, in column $2$ of rows $1$ to $3$, we can observe that the physical projections of the geodesic paths from the classical elastica model are trapped into a shortcut problem, i.e.\ the projected paths travel along non-target regions. The reason is that the ``correct'' path features strongly curved segments whose squared curvature accumulate and lead to a high value of the classical elastica energy, hence it is not selected by this model. In contrast, in column $3$ of the same rows, the curvature prior elastica model indeed leads to accurate extraction of the target structures, even through the two endpoints of the structures in rows $2$ and $3$ are very close to each other. 
In rows $4$ and $5$, each of the curvilinear structures is interrupted by an almost straight segment with slightly stronger appearance features, i.e.\ with higher values of the orientation score along this segment, yielding challenging crossing structures. In these tests, the physical projections of the geodesic paths from the classical elastica model pass through the interrupting segments, yielding a short branches combination problem~\citep{chen2018minimal}, i.e.\ the obtained geodesic paths travel different segments not belonging to the targets. In contrast, the curvature prior elastica model produces satisfactory results which are able to accurately delineate the target curvilinear structures, as shown in column $3$. 

\begin{figure}[t]
\centering
\includegraphics[width=0.33\linewidth]{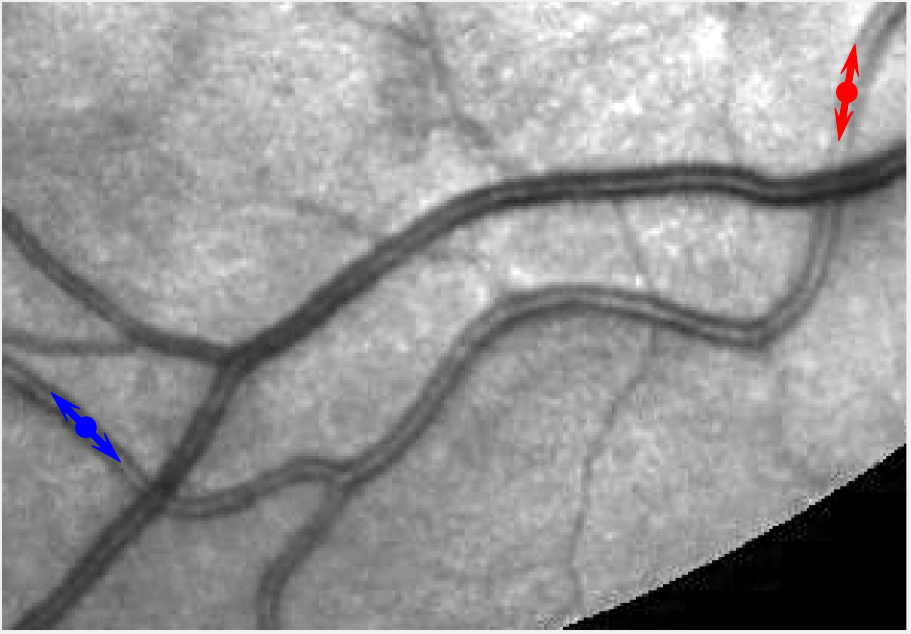}\,\includegraphics[width=0.33\linewidth]{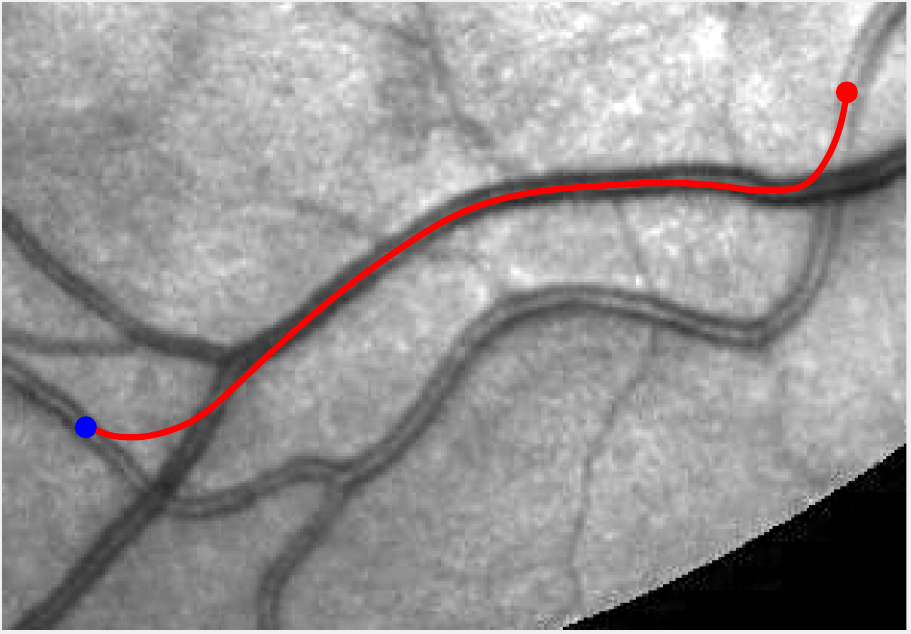}\,\includegraphics[width=0.33\linewidth]{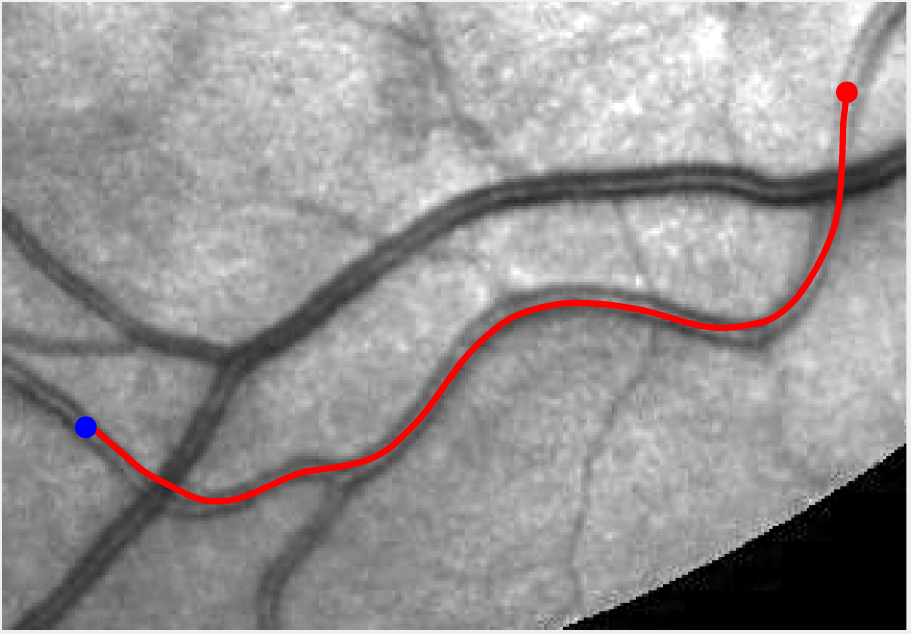}\\
\includegraphics[width=0.33\linewidth]{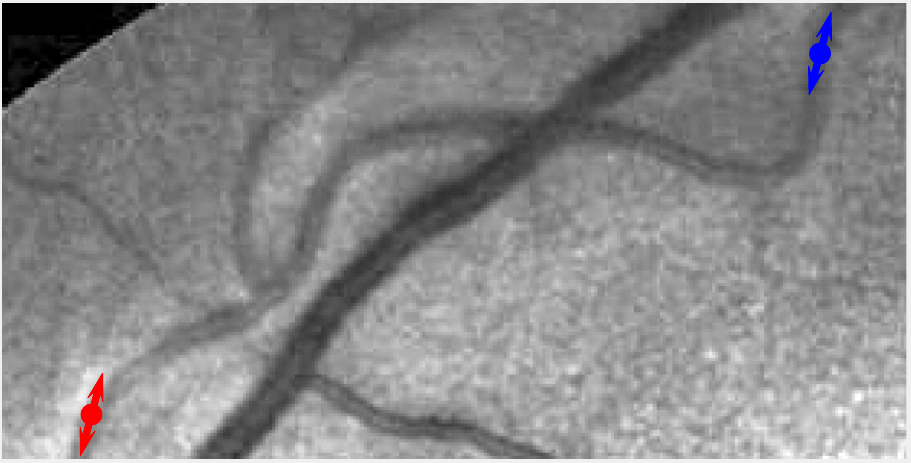}\,\includegraphics[width=0.33\linewidth]{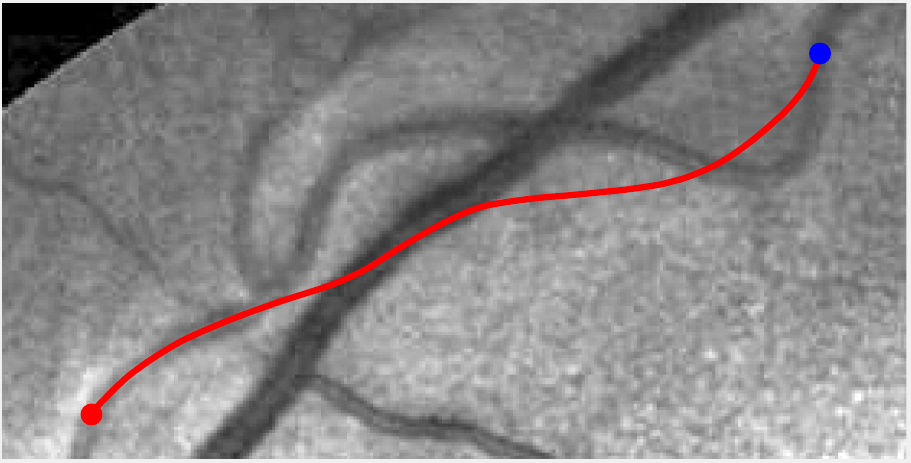}\,\includegraphics[width=0.33\linewidth]{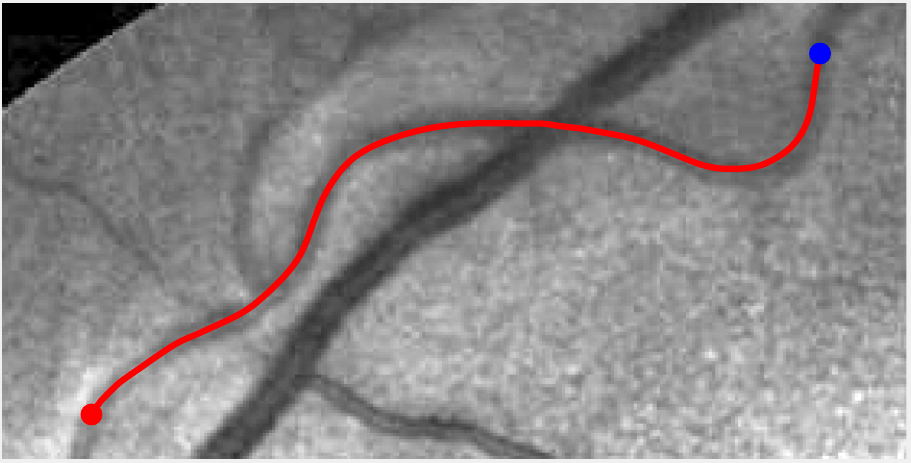}\\
\includegraphics[width=0.33\linewidth]{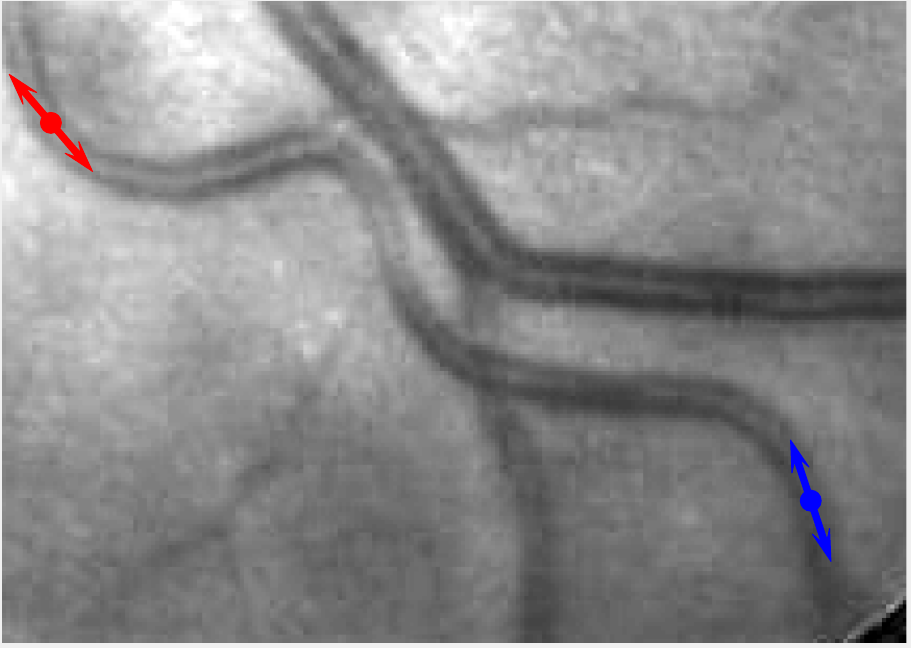}\,\includegraphics[width=0.33\linewidth]{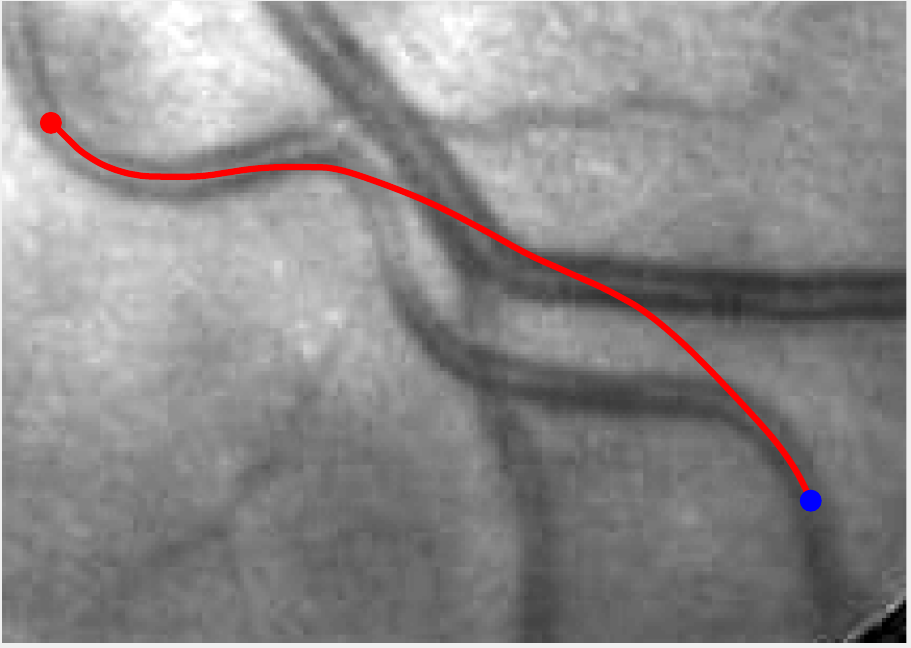}\,\includegraphics[width=0.33\linewidth]{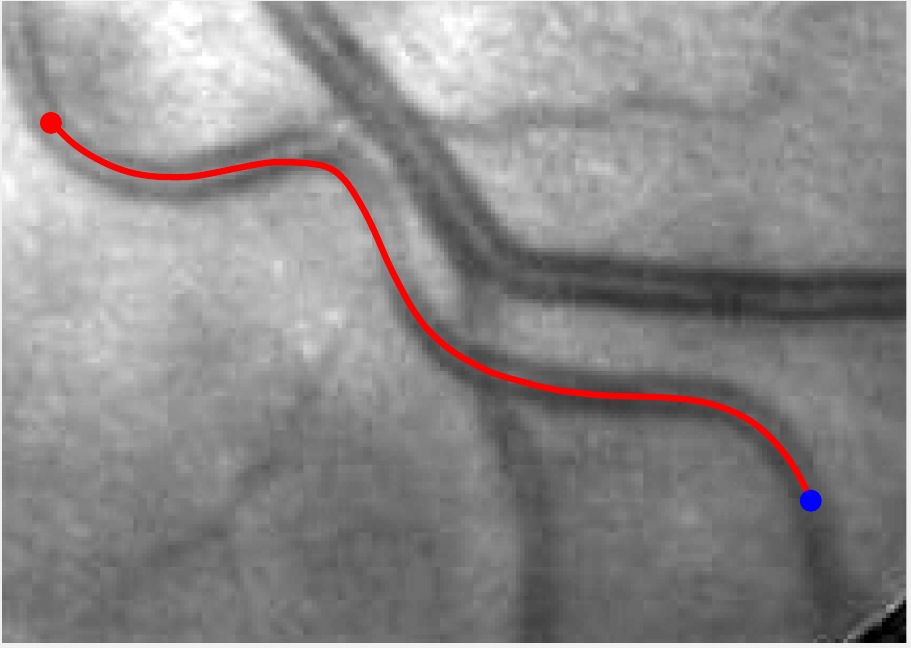}
\caption{Qualitative comparison results on retinal images. \textbf{Left}: The initialization. The blue and red dots with the associated arrows indicate the source and target points. \textbf{Center} and \textbf{Right}: The computed minimal geodesic paths (red lines) for the classical elastica model and the curvature prior elastica model, respectively.}
\label{fig_CompRetina}
\end{figure}

\begin{figure}[t]
\centering
\includegraphics[width=0.33\linewidth]{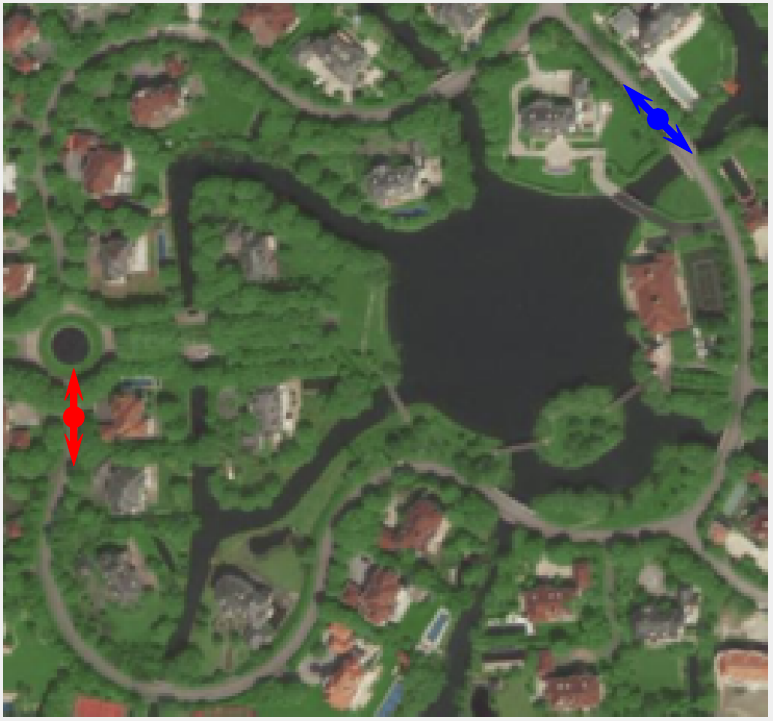}\,\includegraphics[width=0.33\linewidth]{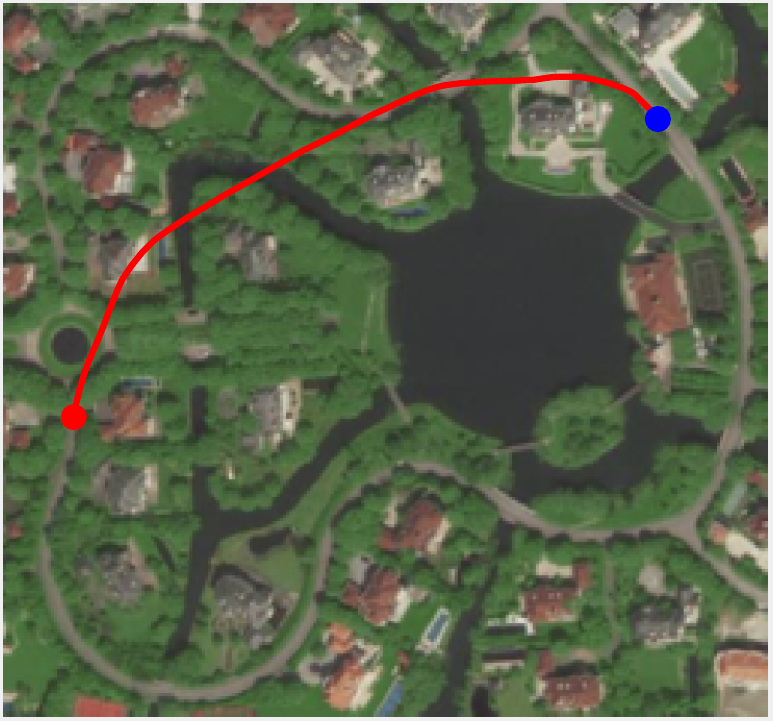}\,\includegraphics[width=0.33\linewidth]{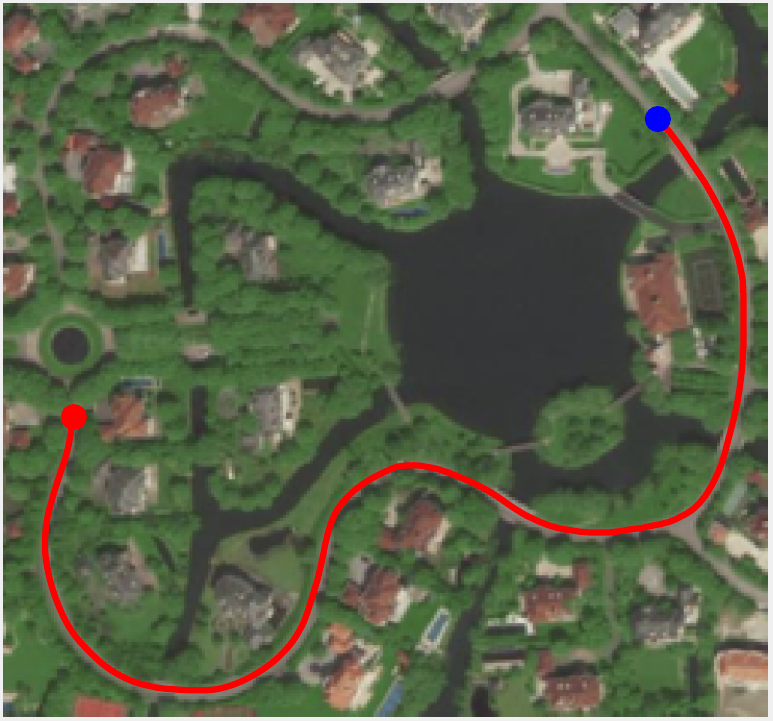}\\
\includegraphics[width=0.33\linewidth]{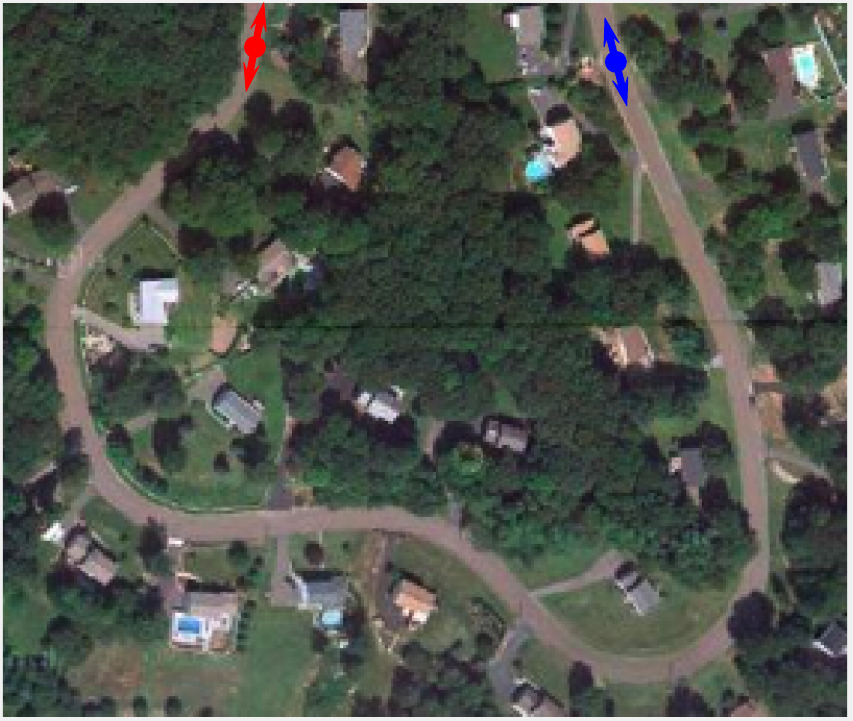}\,\includegraphics[width=0.33\linewidth]{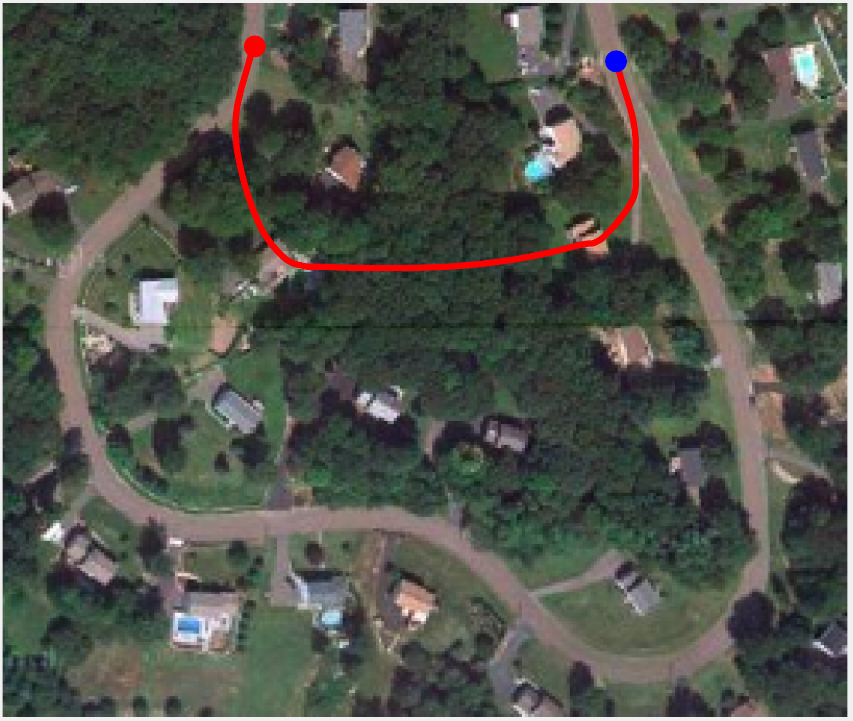}\,\includegraphics[width=0.33\linewidth]{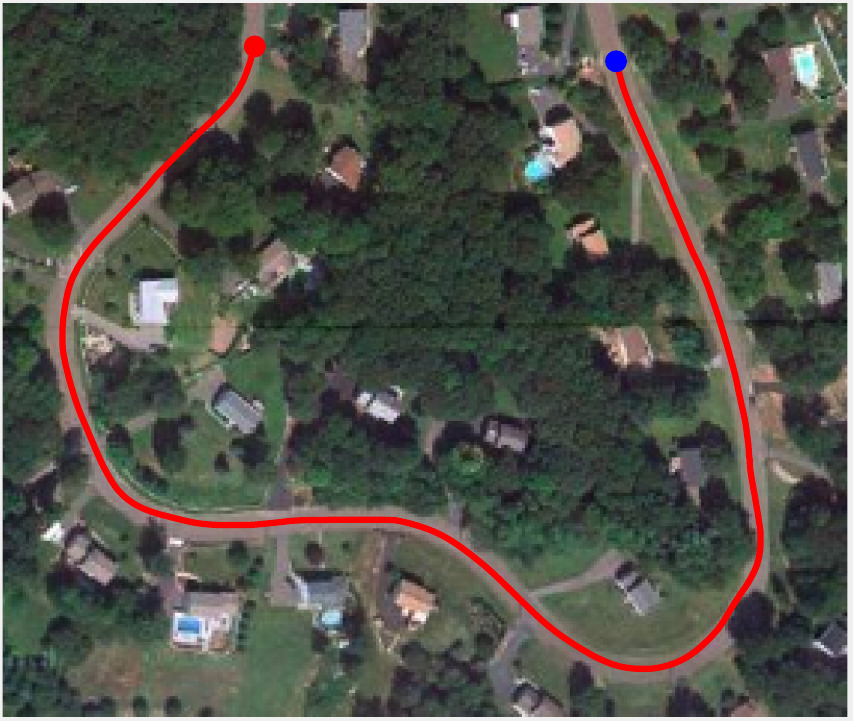}
\caption{Qualitative comparison results on road images. \textbf{Left}: The initialization. The red and cyan dots with the respective arrows indicate the source and target points. \textbf{Center} and \textbf{Right}: The computed minimal geodesic paths (red lines) for the classical elastica model and the curvature prior elastica model, respectively.}
\label{fig_CompRoads}
\end{figure}

Fundus vessel tracking is a task of fundamental importance in retinal imaging and the relevant disease diagnosis. In Fig.~\ref{fig_CompRetina}, we compare the curvature prior elastica model to the classical elastica model in $3$ patches of retinal fundus images, chosen from the IOSTAR dataset~\citep{zhang2016robust}. In this experiment, the goal is to extract artery blood vessels, given two endpoints of each target vessel, as shown in the left column of Fig.~\ref{fig_CompRetina}. In a typical retinal fundus image, the major artery vessels usually exhibit weaker appearance features than those of the nearby vein vessels, and often cross over vein vessels yielding junction structures. An additional difficulty is that the target artery vessels also consist of segments with strong tortuosity. For the classical Euler-Mumford elastica model, one needs to assign a low importance weight to the curvature term, in order to avoid the underlying shortcut problem at those segments. However, this in turn encourages the computed paths to pass through the stronger neighboring vein vessels (whose cost $\psi$ is lower), thus increasing the risk of the short branches combination problem, see the center column of Fig.~\ref{fig_CompRetina}.  In contrast, the minimization of the energy~(\ref{eq_VariantLiftedEnergy}) defining the curvature prior elastica model favors optimal geodesic paths whose curvature is close to $\omega$, allowing to  handle scenarios where challenging structures such as crossings and strongly bent segments appear, see the right column of Fig.~\ref{fig_CompRetina}.

Another practical application for geodesic models is road tracking from aerial images, as demonstrated in Fig.~\ref{fig_CompRoads}. In this experiment, the centerlines of the objective roads have high Euclidean length, and simultaneously feature segments with low turning radii, as shown in Fig.~\ref{fig_CompRoads}a. We can see that the classical elastica model generates tracking results which suffer from a shortcut problem, see the center column of Fig.~\ref{fig_CompRoads}. In contrast, accurate results derived from the curvature prior elastica model are observed in the right column of Fig.~\ref{fig_CompRoads}, thanks to the use of the curvature prior term which alleviates difficulties associated to the complex scenarios.

\begin{table}[t]
\centering
\caption{Quantitive comparison between the classical Euler-Mumford elastica model and the curvature prior elastica model on synthetic images, illustrated by the average and the standard deviation of the Jaccard score values.}
\label{table_HD_Syn}
\setlength{\tabcolsep}{5.5pt}
\renewcommand{\arraystretch}{1.5}
\begin{tabular}{l c c c c c c c c}
\midrule
& &\multicolumn{2}{c}{Elastica Model} & &\multicolumn{2}{c}{Curvature Prior Elastica Model}\\
\cline{3-4}  \cline{6-7} 
&   & Ave. & Std.   & & Ave.  & Std.   \\
\hline
$\beta=4.5$&           &$52.72\%$  &$0.27$   & &$91.83\%$  &$0.03$ \\
$\beta=6.0$&           &$51.47\%$  &$0.26$   & &$90.61\%$  &$0.07$ \\
\bottomrule
\end{tabular}
\end{table} 

\subsection{Quantitative Comparison}
We perform a quantitative evaluation of the performance and stability of the classical elastica model and of the  curvature prior elastica model on synthetic images, when one varies the parameters $\beta$ and $\alpha$ involved in these models, see Eqs.~\eqref{eq_VariantLiftedEnergy} and~\eqref{eq_CostFunction}. These parameters determine the balance between the image data and the curvature regularization. The test data are generated by incorporating different levels of additive Gaussian noise to the original binary images of those shown in Fig.~\ref{fig_CompSyn}. More specifically, we add zero-mean Gaussian noise with $16$ levels of normalized variance values between $0$ and $0.15$ to each clean image, producing $16\times5=80$ test images. Both geodesic models are configured by choosing parameters $\beta\in\{4.5,\,6\}$ and $\alpha\in\{3,4,5\}$, thus yielding $6$ runs per test image.
In each test, the accuracy of the curvilinear structure tracking is measured using the Jaccard score between two bounded regions $A,B\subset\bR^2$, defined as
\begin{equation}
\label{eq_Jaccard}
\JS(A,B):=\frac{|A\cap B|}{|A\cup B|},
\end{equation}
where $|A|$ represents the area of the region $A$. 
In our case, the regions $A$ and $B$ are defined as tubular neighborhoods with a fixed radius of $6$ grid points along the ground truth centerline and the physical projection of the computed geodesic path, respectively. 
Table~\ref{table_HD_Syn} shows the quantitative results comparing the classical elastica model and the curvature prior elastica model, in terms of the Jaccard score between the ground truth centerlines and the  physical projections of geodesic paths, see~\eqref{eq_Jaccard}.
In each row of Table~\ref{table_HD_Syn}, the average (Ave.) and standard deviation (Std.) values of $\JS$, which are derived from $480$ runs w.r.t. different values of the parameters $\alpha$ and $\beta$. Overall, the results from the classical Euler-Mumford elastica model in general exhibit low average values of $\JS$, and large standard deviation, reflecting the fact that various shortcut and short branches combination problems occur in most of the tests. In contrast, we observe much higher average values of $\JS$ using the curvature prior elastica model, and a smaller standard deviation, illustrating its ability to accurately and robustly extract curvilinear structures in the presence of strongly bent segments and junctions.

\section{Conclusion and Discussion}
In this work, we present a nonlinear HJB PDE approach to numerically compute the curve which globally minimizes a variant of the classical Euler-Mumford elastica energy, involving a curvature prior enhancement. The main contributions include 
(i) the establishment of the explicit expression of the Hamiltonian of the curvature prior elastica model, in such way that globally optimal paths can be characterized using the HJB PDE framework,   
(ii) the design of a finite differences scheme based on the HFM method for estimating the numerical solution to the associated HJB PDE, and 
(iii) the introduction of an efficient and robust method for computing the curvature prior in the context of curvilinear structure tracking. Numerical experiments indeed demonstrate promising results in medical and aerial images.

From the modeling standpoint, rather than penalizing the difference to a given curvature prior $\omega$ as proposed here, an alternative interesting route would be to penalize the variation of curvature along the path. 
The latter approach indeed appears to be appropriate and reasonable in relevant applications; for instance it would naturally address the shortcut problem, a segmentation artefact which usually involves sharp turns. 
For that purpose, one may for instance introduce in the elastica energy defined in~\eqref{eq_PriorOriginalEnergy} the squared derivative $\kappa'(\cdot)^2$ of the path curvature. However, this defines a third-order geodesic model, which raises substantial numerical difficulties.
A closely related problem considered in the literature is the computation of geodesic paths with bounded derivative of curvature, however the approach~\cite{bakolas2009curvature_derivative} cannot take into account a data-driven cost function, which is required by image processing applications.
Another related problem is the penalization of torsion, a quantity which is also defined in terms of the third-order derivatives of the path, and is approximated in~\cite{ulen2015torsion} with the torsion of some local Bezier interpolations of the path. In the HJB PDE framework, considered in this paper, the computation of minimal geodesic paths for third-order planar models as considered here would involve the numerical solution of a strongly anisotropic PDE posed on the four dimensional domain $\Omega \times \bS^1 \times \bR \ni (\fx, \theta, \kappa)$, which implies an increased computational cost. We regard it as an opportunity for future research.

We have introduced a practical method for computing the curvature prior $\omega$, using a set of reconstructed piecewise smooth centerlines which reflect the geometry of curvilinear structures in the image. In our work, those centerlines are numerically generated by applying a skeletonization procedure on the segmented curvilinear structures.
A proper integration of the voting method~\cite{rouchdy2013geodesic} and of the keypoints detection method~\cite{benmansour2009fast,kaul2012detecting}  might be an alternative way for constructing the centerlines, where the dependency on the segmentation process can be removed. Furthermore, the estimation of the curvature of curves or surfaces from an orientation score map could also be investigated for  computing $\omega$, e.g.~\cite{bekkers2015curvature}. Unfortunately, the orientation score values around the positions where complex junctions appear usually lack reliability, thus requiring additional procedures to alleviate this problem. 

We have established in~\cref{prop:metric_linear_composition} a general theoretical framework for deriving the Hamiltonian and control set of a curvature-penalized geodesic model equipped with a curvature prior enhancement, and have shown its application to the Euler-Mumford elastica model. This framework can also be exploited to generalize the Reeds-Sheep forward model~\cite{duits2018optimal} and the Dubins car model~\cite{mirebeau2018fast}. 
Finally, we intend to use the curvature-penalized geodesic models with curvature prior enhancement to address various image segmentation problems.

\appendices

\section{Finite Difference Scheme for the Hamiltonian $\kH$}
\label{sec_FDS}
We introduce here the numerical scheme for estimating the geodesic distance map associated to the Hamiltonian $\kH$ based on the state-of-the-art HFM method~\cite{mirebeau2018fast,mirebeau2019riemannian}. Recall that the distance map $\cD_\fs$ is the unique viscosity solution to the static HJB PDE 
\begin{equation}
\label{eq_HJB_biased}
	\kH(\fx,d\cD_\fs(\fx))=\tfrac{1}{2}\psi(\fx)^2, \quad \forall \fx\in\bM\backslash\{\fs\},
\end{equation} 
with the same point source $\fs$ and outflow boundary conditions as in \eqref{eq_HJB}. 
We discretize this PDE over the Cartesian grid
\begin{equation*}
\bM_h:=(h\bZ^2\times (h\bZ/2\pi\bZ))\cap\bM,
\end{equation*}
where $h$ is the grid scale, which is chosen so that $2\pi/h$ is a positive integer for periodicity. 
Note that $h\bZ/2\pi\bZ = \{0, h, \cdots, (N_\theta-1) h\} \subset \bS^1$, where $N_\theta := 2\pi/h$ represents the number of angles discretized in the orientation dimension. We fix $N_\theta=72$ in the numerical experiments of this paper.

Given a non-zero vector $\dfv\in\bE$ and a relaxation parameter $\ve\in\,]0,1[$, a crucial ingredient of the HFM method \cite{mirebeau2018fast} is an approximation of the form
\begin{equation}
\label{eq_DirectionalApprox}
\< \hfx,\dfv\>_+^2=\sum_{k=1}^K\xi_k^\ve(\dfv)\< \hfx,\dfe^\ve_k(\dfv)\>_+^2+\|\hfx\|^2\|\dfv\|^2\cO(\ve^2),
\end{equation}
for any co-vector $\hfx\in\bE^*$, where $\xi^\ve_k(\dfv)\geq 0$ is a non-negative weight, and $\dfe_k^\ve(\dfv)\in\bZ^3$ is an offset having integer coordinates, for each $1\leq k\leq K$ where $K$ is a positive integer. 
The construction of these weights and offsets relies on Selling's decomposition of positive quadratic forms~\cite{mirebeau2014efficient}, with $K=6$ in dimension $\dim(\bE)=3$, and $\|\dfe_k^\ve(\dfv)\| = \cO(\ve^{-1})$. In practice,  we use $\ve=0.1$. 

By applying the approximation scheme in~\eqref{eq_DirectionalApprox} to the terms of \eqref{eq_Approximation}, we obtain the following expression
\begin{equation}
\label{eq_DirectionalApprox2}
\<\hfx,\,\dkq(\fx,\varphi_l)\>_+^2
\approx\sum_{k=1}^K\xi_{k}^\ve(\dkq(\fx,\varphi_l))\<\hfx,\dfe^\ve_{k}(\dkq(\fx,\varphi_l)) \>_+^2.
\end{equation}

For the sake of readability, we leverage the abbreviations $\xi^\ve_{kl}(\fx)$ and $\dfe^\ve_{kl}(\fx)$ to denote $\xi_{k}^\ve(\dkq(\fx,\varphi_l))$ and $\dfe^\ve_{k}(\dkq(\fx,\varphi_l))$, respectively. Combining \eqref{eq_Approximation} with~\eqref{eq_DirectionalApprox2} yields the following approximation of the Hamiltonian $\kH$
\begin{equation}
\kH(\fx,\hfx)\approx\sum_{l=1}^L\mu_l\sum_{k=1}^K\xi_{kl}^\ve(\fx)\left\<\hfx,\dfe_{kl}^\ve (\fx) \right\>_+^2.	
\end{equation}
The set $\cS^\ve(\fx)=\{\dfe^\ve_{kl}(\fx)\}_{1\leq l\leq L}^{1\leq k \leq K} \subset \bZ^3$ of all offsets, at a given discretization point $\fx \in \bM_h$, is known as the stencil of the HFM method and defines the neighborhood system of this numerical scheme, see Fig.~\ref{fig_Stencil}.

\begin{figure*}[t]
\centering
\includegraphics[height=0.4\linewidth]{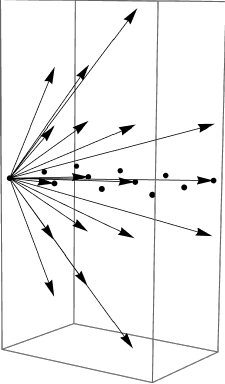}\hspace{10mm}\includegraphics[height=0.4\linewidth]{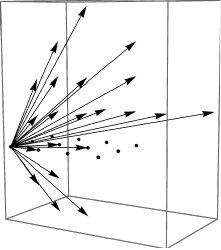}
\caption{Visualizing the stencils associated with the classical elastica model and the curvature prior elastica model. \textbf{Left}: Stencil $\cS^\ve(\fx)$ for the classical Euler-Mumford elastica model with $\beta=2$, $\theta=\pi/6$, and discretization parameters $L=5$, $\ve=0.1$. \textbf{Right}: Likewise for the curvature prior elastica model with $\omega(\fx)=0.5$. Note that the latter is not symmetric  w.r.t.\ the horizontal plane.}
\label{fig_Stencil}
\end{figure*}

Now we can discretize the introduced Hamiltonian $\kH$ using a first order upwind finite difference scheme~\cite{mirebeau2018fast}
\begin{equation}
\label{eq:finite_difference}
\< d\cD_\fs(\fx),\dfe \>_+^2=\left(\frac{\cD_\fs(\fx)-\cD_\fs(\fx-h\dfe)}{h}\right)_+^2+\cO(h),
\end{equation}
for any discretization point $\fx \in \bM_h$ and any offset $\dfe\in\cS^\ve(\fx)$. 
Since the offset $\dfe$ has integer coordinates by construction, one has $\fx-h\dfe \in \bM_h$, and therefore the r.h.s.\ of \eqref{eq:finite_difference} only involves values of the unknown geodesic distance map $\cD_\fs$ on the discretization grid $\bM_h$. 
Finally, we obtain a finite differences discrete counterpart of the HJB PDE~defined in~\eqref{eq_HJB_biased}
\begin{equation} 
\label{eq_DiscreteHJB}
\sum_{l=1}^L\mu_l\sum_{k=1}^K\xi^\ve_{kl}(\fx)\left(\frac{\cD_\fs(\fx)-\cD_\fs(\fx-h\dfe^\ve_{kl}(\fx))}{h}\right)_+^2=\frac{1}{2}\psi(\fx)^2
\end{equation}
for all $\fx \in \bM_h \sm \{\fs\}$, with boundary condition $\cD_\fs(\fs)=0$.

The HFM method computes the numerical solution to the discretized HJB PDE (see~\eqref{eq_DiscreteHJB}) in an efficient single-pass wavefront propagation manner, using a generalization of the original Fast-Marching method~\cite{sethian1999fast}. Assuming that the computation is performed on a finite subset of the Cartesian grid $\bM_h$, containing $N_h$ points, the computation complexity for estimating the solution to~\eqref{eq_DiscreteHJB} is $\cO(L K N_h\ln N_h)$, where $L$ is the discretization parameter as used in~\eqref{eq_Approximation}, and $K = 6$ as observed below \eqref{eq_DirectionalApprox}. In practice, the choice $L=5$ is a good compromise between numerical accuracy and computation cost. Furthermore, the computation time of the HFM method can be deeply speeded up by exploiting a GPU-implemented scheme for estimating the geodesic distances, as proposed in~\citep{mirebeau2023massively}. The codes for the HFM method associated with the variant of the Euler-Mumford elastica model equipped with a curvature prior, involving the construction for the stencils, the computation for the geodesic distance map and the numerical solution to the geodesic backtracking ODE,  can be downloaded from~\href{https://github.com/Mirebeau/HamiltonFastMarching}{https://github.com/Mirebeau/HamiltonFastMarching}.

\bibliographystyle{spbasic}      
\bibliography{minimalPaths}   

\end{document}